\documentclass[12pt]{amsart}

\usepackage{graphicx}
\usepackage{pstricks}

\newtheorem{thm}{Theorem}[section]
\newtheorem{proposition}[thm]{Proposition}

\newtheorem{remark}[thm]{Remark}

\usepackage{amsmath}
\usepackage{amsxtra}
\usepackage{amscd}
\usepackage{amsthm}
\usepackage{amsfonts}
\usepackage{amssymb}
\usepackage{eucal}

\newcommand{\bm}{{\mathbf m}}

\newcommand{\al}{{\alpha}}

\newcommand{\e}{\varepsilon}
\newcommand{\Lo}{{\rm Log}}

\def\Xint#1{\mathchoice
   {\XXint\displaystyle\textstyle{#1}}%
   {\XXint\textstyle\scriptstyle{#1}}%
   {\XXint\scriptstyle\scriptscriptstyle{#1}}%
   {\XXint\scriptscriptstyle\scriptscriptstyle{#1}}%
   \!\int}
\def\XXint#1#2#3{{\setbox0=\hbox{$#1{#2#3}{\int}$}
     \vcenter{\hbox{$#2#3$}}\kern-.5\wd0}}

\def\dashint{\Xint-}

\numberwithin{equation}{section}

\begin{document}
\title{Asymptotic shapes with free boundaries}

\author{Philippe Di Francesco}
\address{Institut de Physique Th\'eorique du Commissariat \`a l'Energie Atomique,
Unit\'e de Recherche associ\'ee du CNRS,
CEA Saclay/IPhT/Bat 774, F-91191 Gif sur Yvette Cedex,
FRANCE. e-mail: philippe.di-francesco@cea.fr}

\author{Nicolai Reshetikhin}
\address{Department of Mathematics, University of California at
Berkeley,
Berkeley, CA 94720-3840. E-mail:
reshetik@math.berkeley.edu}

\date{\today}
\begin{abstract}
We study limit shapes  for dimer models on domains of the hexagonal lattice with free boundary conditions. This is equivalent to the large deviation phenomenon for a random stepped surface over domains fixed only at part of the boundary.
\end{abstract}

\maketitle

\section*{Introduction}

In this paper we study limit shapes for random tilings by rhombi with free boundary conditions.

Recall that there is a local bijection between 3D-partitions, tilings of a plane (of a triangular lattice on a plane) by
rhombi, and dimer configurations on the dual hexagonal lattice
\cite{NHB}.
Another important local bijection is between lattice paths and rhombus tilings.
It is illustrated on Fig.\ref{tilingcuthex}.

The limit shape phenomenon is studied well in dimer
models on hexagonal lattices and for corresponding tiling models
in the case when the boundary of the domain consists of long pieces parallel to the axes of the triangular lattice.
There are also some results for the case when
the boundary is a fixed "zig-zag" path. But in all cases the shape
of the boundary is fixed.

The goal of this paper is to study the limit shapes
when part of the boundary is a random variable. We
call this ``free boundary". An example of a free boundary for
a cut hexagon region is shown on Fig. \ref{tilingcuthex}. The end points
of lattice paths on the Noth-East boundary are not fixed,
which corresponds to variable configurations of boundary  tiles.

In this paper we assume that configurations (tilings, 3D-partitions, paths) are distributed uniformly, or weighted with the probability
\begin{equation}\label{b-weight}
Prob(S)\propto q^{|S|}
\end{equation}
where $q$ is a real positive number, and $|S|$ is the
volume under the height function corresponding to the configuration.

The computations are based on the Gessel-Vionnet formula for the number of lattice paths \cite{GV}. In the last section we compare our results with the analysis of limit shapes of height functions for dimer models based on the Burgers-Hopf equations \cite{KO}.

In Section \ref{one} we compute the boundary of the limit shape
for cut hexagons (see Fig.\ref{tilingcuthex}) of large size when the lattice paths are distributed as in
(\ref{b-weight}). We assume that as the size of the cut hexagon increases, $q$ is approaching $1$
as $q=\exp(\e)$ where ${1\over \e}$ is proportional to the linear size of the system.
In Section \ref{triangles} we
extend these results to the case when there are ``frozen"
triangles along the free boundary. In particular, when two frozen
triangles are next to the ends of the free boundary, the problem
includes the non-cut hexagon as a particular case Fig.\ref{fig:hexa}.
In Section \ref{hexagon} we compare our results for the full hexagon with the paper \cite{CLP} where they were derived originally. Finally Section \ref{burgers} contains the comparison of our results with the derivation of the limit shape from the Hopf-Burgers equation \cite{KO}. In Section \ref{symmetries} we use a symmetry principle
and the results obtained in previous sections to obtain the limit shape for the TSSCPP problem \cite{TSdots}.

The work of N.R. was supported by Danish National
Research Foundation through the
Niels Bohr initiative, by the NSF grant DMS-0601912,
and by DARPA. P.D.F. is supported in part by the ANR Grant GranMa, the
ENIGMA research training network MRTN-CT-2004-5652,
and the ESF program MISGAM.

\section{Case of a cut hexagon}\label{one}

In this section we will focus on random tiling with weights
(\ref{b-weight}). The goal is to find the asymptotic distribution of boundary positions of paths in the continuum limit
i.e. when  $k,n\to \infty$ and $q=\exp(\pm\e)\to 1$ such that $k\e\to \beta$ and $n\e\to \alpha$ are finite.

We will not give a rigorous proof of the existence of the
limit shape, i.e. that the random boundary partition converges in probability to the limit shape function.
We will give reasonable standard arguments, based on the variational principle,
that it does,
and compute the limit shape explicitly.

\begin{figure}
\centering
\includegraphics[width=13.cm]{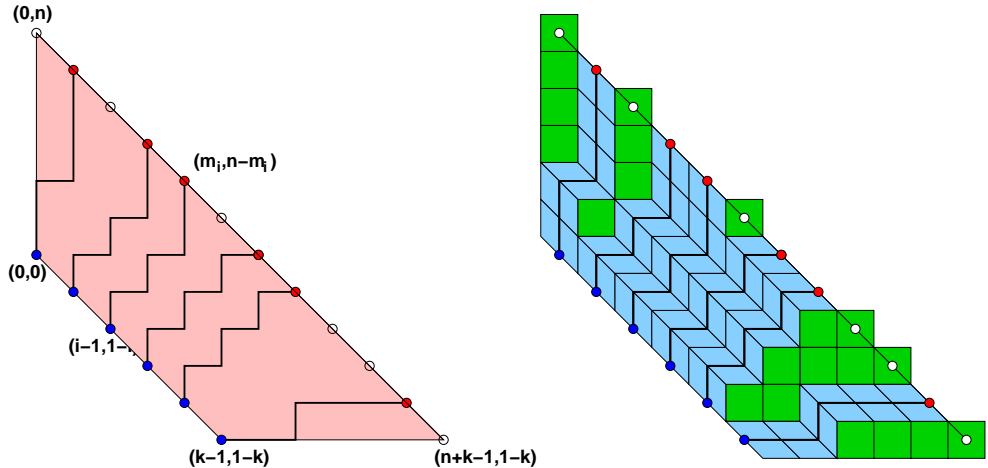}
\caption{A typical tiling of the cut hexagon $n\times k\times n$, with
free boundary conditions on the cut. We have represented the tiling on the square lattice.
One keeps track of the two non-square tiles by drawing a segment in the middle.
The segments form non-intersecting lattice paths, starting at points
$(i-1,1-i)$, $i=1,2,...,k$ and ending on the cut $x+y=n$, at points $(m_i,n-m_i)$,
$0\leq m_i\leq n+k-1$.}\label{tilingcuthex}
\end{figure}

\subsection{ $q$-enumeration}

\begin{figure}
\centering
\includegraphics[width=4.cm]{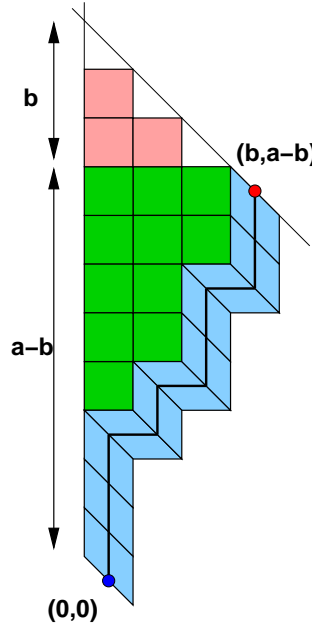}
\caption{The area under the path (green boxes)
plus that in the top triangle (pink boxes)
accounts for the partition function
$q^{b(b-1)/2}{a\choose b}_q$ for lattice paths from $(0,0)$ to $(b,a-b)$,
with a weight $q$ per box in the corresponding slice
of plane partition.}\label{fig:qtile}
\end{figure}

The generating function for the number of tilings
viewed as plane partitions,
with a weight $q$ per box (with green top) can be
expressed as:
\begin{equation}\label{qZ}
Z_{k,n}(\bm\vert q)=\det_{1\leq i,j\leq k}
q^{(m_j-i+1)(m_j-i)\over2} {n \choose m_j-i+1}_q
\end{equation}
where the $q$-binomial reads
${a\choose b}_q= {[a]_q !\over [b]_q! [a-b]_q!}$,
and the $q$-factorial $[a]_q!=\prod_{i=1}^a (1-q^i)$.

This formula is explained in Fig.\ref{fig:qtile}.
The $q$-binomial ${a\choose b}_q$ enumerates the
lattice paths from $(0,0)$ to $(b,a-b)$ with a weight
$q^A$, where $A$ is the (discrete) area under the path,
namely the number of unit squares in the domain
delimited by the axes $x=0$, $y=a-b$ and the path.
This gives the correct weight $q$ per box in the plane partition, within
the slice containing the path, provided one adds up the contribution
of the boxes in the domain delimited by $x=0$, $y=a-b$ and $x+y=a$
(pink boxes), with area ${b \choose 2}$,
hence the additional prefactor $q^{m_j-i+1\choose 2}$.

The formula \eqref{qZ} factors out a Vandermonde type
determinant and it has the following multiplicative form:
\begin{equation}\label{proqZ}
Z_{k,n}(\bm\vert q)=q^{{k+1\choose 3}+{1\over 2}\sum_{i=1}^k m_i(m_i-2k+1)}
\Delta(q^\bm) \,
\prod_{i=1}^k {[n+k-i]_q!\over [m_i]_q! [n+k-m_i-1]_q!}
\end{equation}
where $\Delta(q^\bm)=\prod_{1\leq i,j \leq k}(q^{m_i}-q^{m_j})$.

The generating function is symmetric with respect
to the mapping $q\to q^{-1}$
\begin{equation}\label{q-symm}
Z_{k,n}(\bm\vert q)=q^{s_{k,n}}Z_{k,n}(\bar{\bm}\vert q)
\end{equation}
where $\bar{m}_i=n+k-1-m_i$, $s_{k,n}=-\frac{k^3-k}{2}-\frac{nk(n+k)}{2}+k(n+k-1)$.

\subsection{The continuum limit}

\subsubsection{} Now we consider $\{\bm\}$ as a random variable with
\[
Prob(\{\bm\})\propto Z_{k,n}(\bm\vert q)
\]
and we will study this random variable in the limit when $q=e^{\pm\e}$,
$\e\to +0$, $n,k\to \infty$, and $\alpha=n\e, \beta=k\e$ are fixed.

Let us start with the case $q=e^{-\e}$. The following results can be easily extended to $q=e^\e$ using the symmetry \eqref{q-symm}.

When $\e\to 0$ and $\mu=m\e$ is kept finite we have:
$$
{\rm Log}\, [m]_q! = \frac{1}{\e}\, \int_0^{\mu} {\rm Log}(1-e^{-x})dx+O(1)
$$

Assume that as $\e\to 0$, the density
\[
\mu(t)=\sum_{i=1}^k m_i\e\delta(t-i\e)
\]
converges to the density corresponding to the continuous function $\mu(t)$ on the interval $[0,\beta]$. Then, the leading term of the asymptotic of the generating function \eqref{proqZ} is given by
\[
Z_{k,n}(\bm\vert q)= e^{-{1\over \e^2} S\{\mu\}+o({1\over \e})}
\]
where
\begin{eqnarray*}
S\{\mu\}&=&\int_0^\beta {1\over 2}\mu(t)(\mu(t)-2\beta)
+\int_0^\beta\left[\int_0^{\mu(t)}  {\rm Log}(1-e^{- x})\,dx
+\int_0^{\lambda+1-\mu(t)} \, {\rm Log}(1-e^{- x})\,dx\right]\,dt \\
&& -{1\over 2}\int_0^\beta\int_0^\beta{\rm Log}\vert
e^{-\mu(s)}-e^{-\mu(t)}\vert\,ds\,dt + C(\al,\beta)
\end{eqnarray*}
where $C(\al,\beta)$ reads:
$$
 C(\al,\beta)={1\over 6}-\int_0^\beta dx
\int_0^{\alpha+\beta-x}\, {\rm Log}(1-e^{- y})\,dy
$$

Assume that the function $\mu(t)$ is differentiable and consider the
density $\rho(\mu)$:
\[
\rho(\mu(t))={1\over \mu'(t)}, \ \ dt=\rho(\mu)d\mu
\]
In terms of $\rho(\mu)$ the functional $S$ reads as
\begin{eqnarray*}
S[\rho]&=&\int_0^{\alpha+\beta}({1\over
  2}\mu(\mu-2\beta)+L(\mu)+L(\alpha+\beta-\mu)) \rho(\mu)d\mu \\
&&-{1\over 2}\dashint_0^{\alpha+\beta} \dashint_0^{\alpha+\beta} \ln\vert
e^{-\mu}-e^{-\nu}\vert\rho(\mu)\rho(\nu)\,d\mu\,d\nu
\end{eqnarray*}
where $L(x)=\int_0^x\ln(1-e^{-t})dt$.

In addition, since $\sum_{i=1}^k 1=k$, we have
\begin{equation}\label{den-constr}
\beta=\int_0^\beta dt=\int_0^{\alpha+\beta}\rho(\mu)d\mu
\end{equation}

The functional $S[\rho]$ is convex and therefore has unique minimum.

It is easy to find the variation of this functional:
\begin{multline}
\delta S[\rho]=\int_0^{\alpha+\beta}\left[{1\over
  2}\mu(\mu-2\beta)+L(\mu)+L(\alpha+\beta-\mu)\right]
  \delta\rho(\mu)d\mu \\
  -\dashint_0^{\alpha+\beta}\dashint_0^{\alpha+\beta}
\ln\vert e^{-\al\mu}-e^{-\al\nu}\vert\rho(\nu)d\nu
\delta\rho(\mu)d\mu
\end{multline}
The second variation is positive definite:
\[
\delta^2 S[\rho]=-\dashint_0^{\alpha+\beta}\dashint_0^{\alpha+\beta}
\ln\vert e^{-\mu}-e^{-\nu}\vert\delta\rho(\nu)
\delta\rho(\mu)d\mu d\nu
\]
Therefore the minimum of $S[\rho]$ is achieved at its critical point. Let us find the critical point of $S$ which satisfies the constraint (\ref{den-constr}). The convexity of $S$ guarantees its uniqueness.

The condition (\ref{den-constr}) implies the constraint on
the variation of $\rho$:
\[
\int_0^{\alpha+\beta} \delta\rho(\mu)\, d\mu=0
\]
Taking it into account, we obtain the integral equation
for critical points of $S[\rho]$:
\[
{1\over
2}\mu(\mu-2\beta)+L(\mu)+L(\alpha+\beta-\mu)-\dashint_0^{\alpha+\beta}
\ln\vert e^{-\mu}-e^{-\nu}\vert\rho(\nu)\,d\nu=const
\]
We want to solve it on the space of functions $\rho$ satisfying the constraint (\ref{den-constr}).

Differentiating this integral equation and taking into account (\ref{den-constr}) we arrive at
\[
\ln({1-e^{-\mu}\over e^{-\mu}-e^{-\alpha-\beta}})=
\dashint_0^{\alpha+\beta} {\rho(\nu)\,d(e^{-\nu})\over e^{-\mu}-e^{-\nu}}
\]

After changing variables to $u=e^{-\nu}$ and
$L=e^{-\alpha-\beta}$ the equation for the density $\sigma(u)=\rho(\nu)$ becomes
\begin{equation}\label{int-eqn}
\ln({z-L\over 1-z})=\dashint_L^1{\sigma(u)\,du\over z-u}
\end{equation}
where the function $\sigma$ satisfies the  condition
\begin{equation}\label{cond}
\int_L^1{\sigma(u)du\over u}=\beta
\end{equation}

\subsubsection{} Now let us solve (\ref{int-eqn}) on the surface
of the constraint (\ref{cond})
explicitly.

\begin{proposition}
The unique solution to (\ref{int-eqn}) subject to the constraint
(\ref{cond}) is
\begin{equation}\label{newrho}
\rho(\mu)=\sigma(e^{-\mu})=
{2\over \pi}{\rm Arctan}\left( e^{-\mu}
\sqrt{{\rm sinh}\, {(\mu-A)\over 2}\,  {\rm sinh}\, {(B-\mu)\over 2} \over
{\rm sinh}\, { A\over 2}\,  {\rm sinh}\, { B\over 2} } \right)
\end{equation}
where
\begin{eqnarray}\label{abnew}
e^{-B}&=& e^{-{\al+\beta\over 2}}\left({\rm cosh}\,{\al+\beta\over 2}
-\sqrt{{\rm sinh}\,{\beta\over 2}\, {\rm sinh}\, {2\al+\beta\over 2}}\right)\nonumber \\
e^{-A}&=& e^{-{\al+\beta\over 2}}\left({\rm cosh}\,{\al+\beta\over 2}
+\sqrt{{\rm sinh}\,{\beta\over 2}\, {\rm sinh}\, {2\al+\beta\over 2}}\right)
\end{eqnarray}
\end{proposition}
\begin{proof}
Assume the support of $\sigma$ is the interval $[a,b]\subset
[L,1]$ and consider the function
\[
F(z)=\int_a^b {\sigma(u)\,du\over z-u}
\]
on the complex plane with the branch cut along $[a,b]$. It is a
meromorphic function whose boundary values at the interval $[a,b]$
satisfy the relation
\[
2\ln({z-L\over 1-z})=F(z+i0)+F(z-i0)
\]
and $F(z)=O({1\over z})$ as $|z|\to \infty$.

To solve this equation explicitly consider the function
$G(z)={1\over \sqrt{(z-a)(z-b)}}F(z)$. It is a meromorphic
function on the complex plane with the branch cut along $[a,b]$.
But now at the branch cut we have
\[
G(z+i0)-G(z-i0)=2{1\over \sqrt{(z-a)(z-b)}}\ln({z-L\over 1-z})
\]
and $G(z)=O({1\over z^2})$ as $|z|\to \infty$. From here we
obtain:
\begin{equation}\label{F-int}
F(z)={1\over \pi}\sqrt{(z-a)(z-b)} \int_a^b {1\over
\sqrt{(u-a)(b-u)}} \,
 {\rm Log}{u-L\over 1-u}\,{du \over z-u}\,
\end{equation}

As $|z|\to \infty$ this function behave as
\[
F(z)={1\over \pi}\int_L^1 {1\over \sqrt{(u-a)(b-u)}} \,
 {\rm Log}{u-L\over 1-u}\,du \, + O({1\over z})
\]
This gives the first equation for $a,b$:
\[
\int_L^1 {1\over \sqrt{(u-a)(b-u)}} \,
 {\rm Log}{u-L\over 1-u}\,du=0
\]
It is easy to see that the integral vanishes when $a+b=1+L$.

Changing the integration variable in (\ref{F-int}) $u=(b-a)s+a$
and taking into account the relation between $a$ and $b$ we obtain the
following integral formula for $F(z)$:
\[
F(z)={1\over \pi}\sqrt{(z-a)(z-b)} \int_0^1 {1\over
\sqrt{s(1-s)}} \,
 {\rm Log}{s+M\over 1+M-s}\,{ds \over t-s}\,
\]
where $t={z-a\over b-a}$, and $M={a-L\over 1+L-2a}$.

This integral can be computed explicitly using
\[
\int_0^1{1\over (w-u)(\beta+u)\sqrt{u(1-u)}}du={\pi\over \beta+w}\left({1\over \sqrt{\beta(\beta+1)}}+{1\over \sqrt{w(w-1)}}\right)
\]
and we obtain:
\[
F(z)= \ln{z-L\over 1-z}+2i{\rm Arctan}\sqrt{(z-a)(b-z)\over
  (a-L)(1-a)}
\]

The second equation determining $a$ and $b$ is the normalization
(\ref{den-constr}). It gives $F(0)=-\beta$ and therefore:
\[
\sqrt{(1-a)(1-b)\over ab}={\rm tanh} {\alpha\over 2}
\]

This gives $a$ and $b$:
\begin{eqnarray}
a&=& e^{-{\al+\beta\over 2}}\left({\rm cosh}\,{\al+\beta\over 2}
-\sqrt{{\rm sinh}\,{\beta\over 2}\, {\rm sinh}\, {2\al+\beta\over 2}}\right)\nonumber \\
b&=& e^{-{\al+\beta\over 2}}\left({\rm cosh}\,{\al+\beta\over 2}
+\sqrt{{\rm sinh}\,{\beta\over 2}\, {\rm sinh}\, {2\al+\beta\over 2}}\right)
\end{eqnarray}

The density $\rho(\mu)$ can be easily derived from $F(z)$
as $\rho(\mu)=Im F(e^{-\mu})$ when $\mu\in [B,A]$, and  $B=-{\rm Log}\, b$, $A=-{\rm Log}\, a$:
\begin{equation}
\rho(\mu)=\sigma(e^{-\mu})=
{2\over \pi}{\rm Arctan}\left( e^{-\mu}
\sqrt{{\rm sinh}\, {(\mu-A)\over 2}\,  {\rm sinh}\, {(B-\mu)\over 2} \over
{\rm sinh}\, { A\over 2}\,  {\rm sinh}\, { B\over 2} } \right)
\end{equation}
\end{proof}

Because the solution $\rho$ described above is a local minimum
of the convex functional $S$. it is also a global minimum.

The density $\rho(\mu)$ is plotted on Fig.\ref{fig:densitq} for $\lambda=1/2,1,2,5$
and $\beta=0,1/2,1,2,5$. Here $\lambda=\alpha/\beta$.

\begin{figure}
\centering
\includegraphics[width=10.cm]{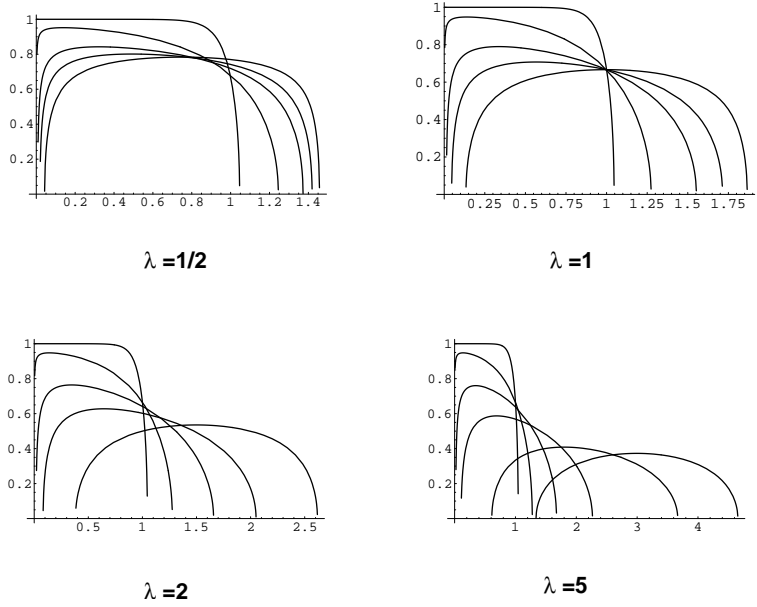}
\caption{Distribution of endpoints for $\lambda=\al/\beta
=1/2,1,2,5$ and
$\beta=0,1/2,1,2,5$ (from right to left, in each plot).}\label{fig:densitq}
\end{figure}

\subsubsection{}
The limit when $\al=\beta\lambda$ and $\al\to 0$ while $\lambda$ is fixed correspond to the limit shape for
uniformly weighted lattice paths, i.e to $q=1$.

In this limit $A=\beta A_0$, $B=\beta B_0$ where

\begin{eqnarray}\label{uni-bound}
B_0(\lambda)&=&{1+\lambda \over 2} -\sqrt{{1\over 2}(1+{\lambda\over 2})}\\
A_0(\lambda)&=&{1+\lambda \over 2} +\sqrt{{1\over 2}(1+{\lambda\over 2})}
\end{eqnarray}
The density of endpoints on the free boundary converges to
\begin{equation}\label{uni-densit}
\rho(t)={2\over \pi}{\rm Arctan}\left(\sqrt{{(t-A_0)(B_0-t)\over A_0B_0}}\right)={2\over \pi}{\rm Arctan}\left({\sqrt{2\lambda+1-4(t-{\lambda +1\over 2})^2}\over \lambda}\right)
\end{equation}
here $t=\mu/\beta$.

\subsubsection{}
Another interesting limit is when $\al\to \infty$ while $\beta$ is fixed. In this case $\rho$ is given by the same formula
as before but for the boundary points we have:
\[
e^{-B}\to {1\over 2} (1-\sqrt{1-e^{-\beta}})
\]
\[
e^{-A}\to {1\over 2} (1+\sqrt{1-e^{-\beta}})
\]

Note finally that all the above is valid for $\al<0$ as well.
The transformation $\al\to-\al$ amounts to reflecting the plane partition w.r.t. the axis $y=x$ as in (\ref{q-symm}).

\section{Case of a cut hexagon with forbidden intervals along the free boundary}\label{triangles}

In the previous section we studied boundary partition for
the distribution $q^{|S|}$. From now on we will consider the
uniform distribution. All results can be easily generalized to
$q\neq 1$  but we will not do it here.

\subsection{Enumeration}

\begin{figure}
\centering
\includegraphics[width=9.cm]{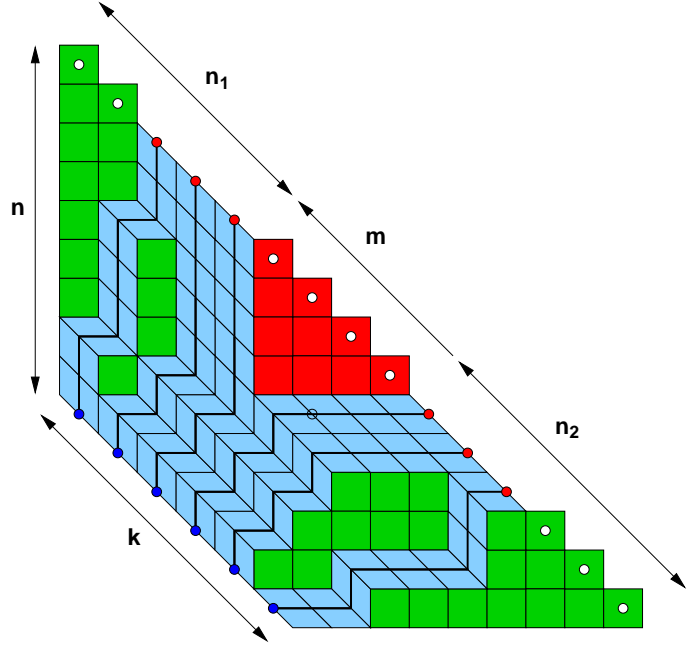}
\caption{The domain to be tiled is a half-hexagon $n\times k\times n$
with the fixed interval of $m$ square tiles along the boundary.}\label{fig:tricut}
\end{figure}

We consider the domain of Fig.\ref{fig:tricut}, made of a cut hexagon $n\times k\times n$
with a fixed interval of square tiles along the boundary. The effect of these tiles is that lattice paths are forbidden to have endpoints along a segment of the form $\{(i,n-i),\ i=n_1,n_1+1,\ldots n_1+m-1\}$, while
$n_1+m+n_2=n+k$.

It is clear that the counting formula is the same
Gessel-Viennot formula as in the previous problem with $q=1$:
\begin{equation}\label{unif-count}
Z_{k,n}(\bm)=\det_{1\leq i,j\leq k}
 {n \choose m_j-i+1}
\end{equation}
It can be written in the multiplicative form:

\begin{equation}\label{proqZprime}
Z_{k,n}(\bm)=\prod_{1\leq i,j \leq k}(m_i-m_j) \,
\prod_{i=1}^k {[n+k-i]!\over [m_i]! [n+k-m_i-1]!}
\end{equation}

The difference with the previous case is that positions
of end points of lattice paths $m_i$ now can only be in
$[0,n+k-1]\setminus [n_1,n_1+m-1]=[0,n_1-1]\cup [n_1+m,n+k-1]$.
The interval $[n_1,n_1+m-1]$ is forbidden.

\subsection{Large size asymptotics}
Here we will study the distribution of endpoints in the limit when $k\to \infty$ and $n=\lambda k$, $n_1=\lambda_1 k$, $n_1+m=\lambda_2 k$, $m_i=\mu_i k$, for some finite $\lambda,\lambda_1$, and $\lambda_2$.

The large $k$ asymptotics of the number of lattice paths with endpoints
fixed outside of the forbidden intervals is given by the action
function

\begin{eqnarray}\label{action-uni}
S_\lambda[\rho]&=&\int_0^{\lambda+1}({1\over
  2}\mu(\mu-2\lambda)+L_0(\mu)+L_0(\lambda+1-\mu)) \rho(\mu)d\mu \\
&&-{1\over 2}\dashint_0^{\lambda+1} \dashint_0^{\lambda+1} \ln\vert
\mu-\nu\vert\rho(\mu)\rho(\nu)\,d\mu\,d\nu
\end{eqnarray}
where $L_0(x)=x\Lo(x)-x$.

This is a convex functional. The same arguments as
in the previous section shows that its extremum is the
unique minimum. Assume that the forbidden segment
$[\lambda_1,\lambda_2]$ is in a generic position which means
that the region where the density is not constant
(not $0$ or $1$) consists of two segments $T=[a,b]\cup[c,d]$, such that $0<a<b< \lambda_1<\lambda_2<c<d<\lambda+1$. The density is
constant outside of $T$:
\[
\rho(\mu)=0, \ \ \mu\in [0,a]\cup[d,\lambda +1], \ \ \rho(\mu)=1, \ \ \mu\in [b,\lambda_1]\cup [\lambda_2,c]
\]

The density $\rho(\mu)$ which minimizes $S[\rho]$
satisfies the integral equation
\begin{equation}\label{cut-hex-2t}
{\rm Log}{u(b-u)(u-\lambda_2)\over (\lambda+1-u)(\lambda_1-u)(u-c)}=
\dashint_T {\rho(\mu)d\mu\over u-\mu}
\end{equation}
$u\in T$ with the additional constraint:
\[
\int_T\rho(u)du+(\lambda_1-b)+(c-\lambda_2)=1
\]

The solution to the integral equation \eqref{cut-hex-2t} with
the additional constraints can be obtained in the
similar way as before by converting the integral equation to
the problem of finding the holomorphic function with
given boundary values. For the resolvent
\[
F(z)=\int_T{\rho(u)du\over z-u}
\]
we obtain the following integral expression
\begin{eqnarray}\label{one-int-F}
F(z)&=&{1\over \pi}\sqrt{(z-a)(z-b)(z-c)(z-d)}\\
&&\times \, \int_T {du \over \sqrt{(b-u)(u-a)(c-u)(d-u)}}\,
{1\over z-u}\, {\rm Log}\,{u(b-u)(u-\lambda_2)\over (\lambda+1-u)(\lambda_1-u)(u-c)}
\end{eqnarray}

The end points $a,b,c,d$ are determined by the asymptotic
behavior $F(z)= (1+\lambda_2-\lambda_1+b-c)/z+O(1/z^2)$ at large $z$.
Indeed, if there were no conditions on $a,b,c,d$ the resolvent would
behave as $z^2$ for large $z$. Vanishing coefficients
in $z^2$, $z$, $z^0$, and fixing the coefficient in $z^{-1}$ give
four equations. If the parameters $\lambda, \lambda_1,\lambda_2$
are in generic position, these equations have a real
solution, which gives the desired values of $a,b,c,d$
as functions of $\lambda, \lambda_1,\lambda_2$.

The density is given by the imaginary part of the boundary value for the resolvent on $T$:
\begin{equation}\label{den-one-int}
\rho(z)={\rm Im}\, F(z+i0)/\pi, \ \ z\in T
\end{equation}

For sufficiently small values of $\lambda_1$ the interval $[a,b]$ may
shrink to a point. The same can happen with the interval $[c,d]$ for
sufficiently large values of $\lambda_2$. Both intervals shrink to a
point when $\lambda_2-\lambda_1=\lambda$. It is impossible to have
$\lambda_2-\lambda_1>\lambda$  because in this case the number of
lattice paths incoming to the region would be greater then the number
of paths outgoing from it.

\subsection{Case of many forbidden intervals}

This is a straightforward generalization of the previous case to a
succession of forbidden intervals along the free boundary.

The number of lattice paths is still given by the formula
\eqref{unif-count}. The only difference is that all
$m_i$ should be outside of the forbidden intervals.

Denote by $[\al_i,\beta_i]$ the forbidden intervals $i=1,\dots, l$.
The density of the outgoing paths is supported on
$\cup_{i=1}^{l+1}[\beta_{i-1},\al_i]$ where $\beta_0=0$ and $\al_{l+1}=\lambda+1$.
Let us say $\al_i,\beta_i$ are generic when the minimizer of the rate functional $S[\rho]$
is constant on the following intervals:
\[
\rho(u)=1, u\in \cup_{i=2}^{l+1} [\beta_{i-1}, a_i]\cup [b_{i-1},\al_{i-1}], \ \
\rho(u)=0, u\in [0,a_1]\cup [b_{l+1},\lambda +1]
\]

The integral equation for the minimizer is
\[
{\rm Log}\left( {\mu \over \lambda+1-\mu}\prod_{i=2}^{l+1} {(\mu-b_{i-1})(\mu-\beta_{i-1})\over
(\al_{i-1}-\mu)(a_i-\mu)}\right)
=\dashint_T dw{\rho(w)\over \mu-w}
\]
where $T=\cup_i [a_i,b_i]$.

The minimizer satisfies the constraint (the number of incoming paths is equal to
the number of outgoing paths):
\[
\int_T\rho(u)du+\sum_{i=2}^{l+1}(a_i-\beta_{i-1})+
\sum_{i=1}^{l}(\al_i-b_{i})=1
\]

The solution to this integral equation is given by the
the resolvent function
$$
F(z)={1\over \pi}
\int_T du\prod_i{\sqrt{(z-a_i)(z-b_i)} \over \sqrt{(b_i-u)(u-a_i)}}\,
{1\over z-u}\, {\rm Log}\,{u\over \lambda+1-u}\prod_i {(u-b_{i-1})(u-\beta_{i-1})\over
(\al_{i-1}-u)(a_i-u)}
$$
as usual:
\[
\rho(z)={\rm Im}\, F(z+i0)/\pi, \ \ z\in T
\]

\subsection{The case of forbidden and fully packed intervals}
\subsubsection{}
Forbidden intervals are just one way to partially enforce boundary
conditions along segments of the free boundary, where one imposes $\rho=0$.
Another special case is that of {\it fully packed} intervals, along which every point is an end point of a lattice path,
hence $\rho=1$ is imposed.

We cannot enforce a given type of tiles along a fully packed interval.
However the geometry of tiles is such that each fully packed interval has a sequence of non-square tiles which start with
horizontally tilted tiles at the top and then horizontal tiles
turn at some point into vertical non-square tiles. The position of the turning
point is random: it is not fixed by the requirement that the interval is fully packed.

The counting formula for the number of lattice paths is the
same as before and we can compute the limit density of
end points of lattice paths along the free part of the boundary
similarly to the case when we have only forbidden intervals.

\subsubsection{}The problem of finding the limit density for one fully packed interval is very similar to the problem for one forbidden interval.
Lattice paths come to the fully packed interval horizontally
at the upper part of the interval and vertically at the lower
part of the interval. Such a configuration induces forbidden intervals at both ends of the fully packed interval.
Let $l_1$ and $l_2$ be end-points of the fully packed interval.
In the limit when $k,n, l_1, l_2\to \infty$ such that $\lambda=n/k, \lambda_i=l_i/k$ are finite the resulting density of lattice paths has the following structure
\[
\rho(\mu)=0, \ \ \mu\in [0,a]\cup[d,\lambda +1], \ \ \rho(\mu)=0, \ \ \mu\in [b,\lambda_1]\cup [\lambda_2,c]
\]
Here $0<a<b<\lambda_1<\lambda_2<c<d<1+\lambda$, $\rho(\mu)=1$
when $\mu\in [\lambda_1,\lambda_2]$ and $\rho(\mu)$ is a smooth
functions when $\mu\in T=[a,b]\cup[c,d]$. It satisfies the
constraint
\[
\int_T\rho(\mu)d\mu +(\lambda_2-\lambda_1)=1
\]
The limit density on $T$ is determined by the variational problem similar to the one described above and the solution is
given by the resolvent function eq. \eqref{one-int-F}.
The density is equal to the imaginary part of the resolvent
on $T$ eq. \eqref{den-one-int}.
The endpoints $a,b,c,d$ are determined by the condition
$F(z)=(1+\lambda_1-\lambda_2)/z+O(1/z^2)$.

When a fully packed interval is near any boundary of the whole interval $[0, \lambda+1]$, the lattice paths
form a ``fully packed region" along the corresponding side of the
cut hexagon. In this case the problem of finding the limit density of paths is equivalent 
to the problem for a cut hexagon of smaller size.

\subsubsection{} The generic case of
several fully packed and forbidden intervals corresponds to an alternance
of the intervals, apart from one-another. All
other situations are degenerations of this one.

The equation for the density in this case is similar to the
case when we have several forbidden intervals. Let $[\alpha_i,
\beta_i]$ be forbidden intervals and $[\lambda_i,\nu_i]$ be fully packed intervals. Assume that $\beta_i<\lambda_i$ for all $i=1,\dots, n$. When the boundaries  of intervals are in generic position the density will vary smoothly in intervals $[a_i, b_i]$,
$[c_i,d_i]$, $i=1,\dots, n$ where $\dots< b_i<\alpha_i<\beta_i<c_i<d_i<\lambda_i<\mu_i<a_{i+1}<b_{i+1}<\dots$,
and $d_n<\lambda+1$. The density is $1$ in fully packed intervals and $0$ in the forbidden intervals. In addition in is $1$ in intervals $[b_i,\alpha_i]$, $[\beta_i, c_i]$, and it is $0$
in the intervals $[d_i,\lambda_i]$, $[\nu_i, a_{i+1}]$, $[0,a_1]$, and $[d_n,\lambda+1]$.

The density function satisfies the constraint:
\[
\sum_{i=1}^n \int_T\rho(\mu)d\mu
+\sum_{i=1}^n((\alpha_i-b_i)+(c_i-\beta_i)+(\nu_i-\lambda_i))=1
\]
where $T=\cup_{i=1}^n [a_i, b_i]\cup[c_i,d_i]$
The solution to the integral equation for the minimizer is given by the resolvent:
\begin{eqnarray*}
F(z)&=&{1\over \pi}
\int_T du\prod_{i=1}^n{\sqrt{(z-a_i)(z-b_i)(z-c_i)(z-d_i)} \over \sqrt{(b_i-u)(u-a_i)(u-c_i)(u-d_i)}} \\
&&\ \ \ \ \ \ \times \,{1\over z-u}\, {\rm Log}\,{u\over \lambda+1-u}\prod_{i=1}^n {(u-b_{i})(u-\beta_{i})(u-\lambda_i)\over
(u-\alpha_i)(u-c_i)(u-\mu_i)}
\end{eqnarray*}
The constraint on the density is equivalent to the following condition on the
asymptotics of $F(z)$:
\[
F(z)=(1-\sum_{i=1}^n((\alpha_i-b_i)+(c_i-\beta_i)+(\nu_i-\lambda_i)))/z+O(1/z^2)
\]
which gives equations defining $a_i,b_i,c_i,d_i$.
The density is given by the imaginary part of $F(z)$
on $T$ as in eq. \eqref{den-one-int}.

\subsection{Case of two forbidden intervals at the corners}
Here we consider the case of two forbidden intervals that touch respectively the top and the bottom
of the free boundary. We will see that in this case the limit density can be computed explicitly in terms elementary functions.

Assume that the upper interval has length $m$ and that the lower one is of length $p$.
Now the values $m_i\in [0,m]$ and $m_i\in [n+k-p-1,n+k-1]$ are
forbidden, and hence
the allowed range for endpoints of lattice paths is  $[m+1,n+k-p-2]$.

Now we will take $k\to \infty$ such that
$n=\lambda k$, $m=\nu k$, $n+k-p=\theta k$, which finite $\lambda,
\nu, \theta$ and will compute the minimizer of the rate function
$S[\rho]$.

\subsubsection{}

For generic $\nu, \theta$ the minimizer of the rate functional
is not constant on the interval $[a,b]$ with $\nu<a<b<\theta$,
and
\[
\rho(u)=1, u\in [\nu,a]\cup [b,\theta]
\]

It satisfies the integral equation
\begin{equation}\label{2FI-eqn}
{\rm Log}\left( {\mu (\mu-a)(\theta-\mu)\over (\lambda+1-\mu)(\mu-\nu)(b-\mu)}\right)
=\dashint_a^b dw{\rho(w)\over \mu-w}
\end{equation}
and the constraint
\begin{equation}\label{norhod}
\int_a^b\rho(w)dw=1-(a-\nu)-(\theta-b)
\end{equation}

In this case the integrals can be computed explicitly.
\begin{proposition}
The density is given by
\begin{eqnarray}\label{2FI-dens}
&&\rho(z)=1+{2\over \pi} \left(  {\rm Arctan}\, \sqrt{b-z\over z-a}\sqrt{a-\nu\over b-\nu}
- {\rm Arctan}\, \sqrt{b-z\over z-a}\sqrt{a\over b} \right. \label{densitwotrip} \\
&&\qquad\qquad \left.
+{\rm Arctan}\, \sqrt{z-a\over b-z}\sqrt{\theta-b\over \theta-a}
- {\rm Arctan}\, \sqrt{z-a\over b-z}\sqrt{\lambda+1-b\over \lambda+1-a}\, \right)\nonumber \\
\end{eqnarray}
with
\begin{equation}\label{triab}
a={1\over 4\lambda^2}(U(\lambda,\nu,\theta)-\sqrt{V(\lambda,\nu,\theta)}) \qquad
b={1\over 4\lambda^2}(U(\lambda,\nu,\theta)+\sqrt{V(\lambda,\nu,\theta)})
\end{equation}
where $U$ and $V$ are the following polynomials:
\begin{eqnarray*}
U(\lambda,\nu,\theta)&=&  (\nu+\theta)\Big(1-(\nu-\theta)^2\Big) -(\lambda+1)
\Big((1-\nu-\theta)(1+2\lambda+\nu+\theta)+4 \nu \theta\Big)\\
V(\lambda,\nu,\theta)&=&
(1-\nu-\theta)(1-\nu+\theta)(1+\nu-\theta)\\
&&\qquad \times (1+\nu-\theta+2\lambda)(1-\nu+\theta+2\lambda)(1-\nu-\theta+2\lambda)
\end{eqnarray*}
\end{proposition}

\begin{proof} The resolvent for \eqref{2FI-eqn} is given by the integral
$$
F(z)={\sqrt{w(w-1)}\over \pi}\int_0^1 {dv\over \sqrt{v(1-v)}}\, {1\over w-v} \,
{\rm Log}\left({(\beta+v)v(\eta-v)\over (\gamma-v)(\delta+v)(1-v)}\right)
$$
where $w={z-a\over b-a}$ and
\begin{equation}\label{valstrid}
\beta={a\over b-a}, \quad \gamma={\lambda+1-a\over b-a},\quad \delta={a-\nu\over b-a},
\quad \eta={\theta-a\over b-a}
\end{equation}

The normalization condition \eqref{norhod} for the density is equivalent to large $z$ asymptotic behavior for the resolvent
$F(z)={1+b-a+\nu-\theta\over z}+O({1\over z^2})$. Vanishing of the coefficient in $z$ and the constant term of this asymptotic  amounts to:
\begin{eqnarray}\label{van-1}
0=\int_0^1 {dv\over \sqrt{v(1-v)}}\,
{\rm Log}\left({(\beta+v)v(\eta-v)\over (\gamma-v)(\delta+v)(1-v)}\right)
\end{eqnarray}
\begin{eqnarray}\label{van-2}
1-(\delta+\eta)(b-a)={b-a\over \pi} \int_0^1  {dv\over \sqrt{v(1-v)}}\, v\,
{\rm Log}\left({(\beta+v)v(\eta-v)\over (\gamma-v)(\delta+v)(1-v)}\right)
\end{eqnarray}
Using identities
\begin{equation}\label{firstint}
\int_0^1{1\over \sqrt{u(1-u)}}{\rm Log}(\beta+u) du=2\pi{\rm Log}\left({\sqrt{\beta}+\sqrt{\beta+1}\over 2}\right)
\end{equation}
\begin{equation}\label{secondint}
\int_0^1{1\over \sqrt{u(1-u)}}u{\rm Log}(\beta+u) du={\pi\over 2}\left((\sqrt{\beta}-\sqrt{\beta+1})^2+2{\rm Log}\left({\sqrt{\beta}+\sqrt{\beta+1}\over 2}\right)\right)
\end{equation}
we obtain equations
\begin{eqnarray*}
0&=&\Big(\sqrt{\beta}+\sqrt{\beta+1}\Big)\Big(\sqrt{\eta}+\sqrt{\eta-1}\Big)
-\Big(\sqrt{\gamma}+\sqrt{\gamma-1}\Big)\Big(\sqrt{\delta}+\sqrt{\delta+1}\Big)\\
1&=& {b-a\over 2}\Big\{2(\delta+\eta+1)
+ (\sqrt{\beta+1}-\sqrt{\beta})^2+(\sqrt{\gamma-1}-\sqrt{\gamma})^2\\
&&\qquad \qquad\qquad\qquad\qquad
-(\sqrt{\delta+1}-\sqrt{\delta})^2-(\sqrt{\eta-1}-\sqrt{\eta})^2\Big\}
\end{eqnarray*}

The integral defining the resolvent can be evaluated explicitly
using the following idenity
\begin{eqnarray}
&&{\sqrt{w(w-1)}\over \pi}\int_0^1 {dv\over \sqrt{v(1-v)}} \, {1\over w-v}\, {\rm Log}\,{A+v\over B+v}=\nonumber \\
&&\qquad \qquad
 {\rm Log}\,{A+w\over B+w}+2i \left({\rm Arctan}\, \sqrt{1-w\over w}\sqrt{B\over 1+B}-
 {\rm Arctan}\, \sqrt{1-w\over w}\sqrt{A\over 1+A}\right) \label{ABint}
\end{eqnarray}
It results in the formula:
\begin{eqnarray*}
&&F(z)={\rm Log}\left({(\beta+w)w(\eta-w)\over (\gamma-w)(\delta+w)(1-w)}\right)
+2i\left( {\rm Arctan}\, \sqrt{1-w\over w}\sqrt{\delta\over 1+\delta} \right.\\
&&\left.
- {\rm Arctan}\, \sqrt{1-w\over w}\sqrt{\beta\over 1+\beta}
+{\rm Arctan}\, \sqrt{w\over 1-w}\sqrt{\eta-1\over \eta}
- {\rm Arctan}\, \sqrt{w\over 1-w}\sqrt{\gamma-1\over \gamma}\, \right)
\end{eqnarray*}

Using this expression for the resolvent, vanishing conditions
\eqref{van-1}\eqref{van-2} can be written as:
\begin{equation}\label{triabi}
a={1\over 4\lambda^2}(U(\lambda,\nu,\theta)-\sqrt{V(\lambda,\nu,\theta)}) \qquad
b={1\over 4\lambda^2}(U(\lambda,\nu,\theta)+\sqrt{V(\lambda,\nu,\theta)})
\end{equation}
where $U$ and $V$ are the following polynomials:
\begin{eqnarray*}
U(\lambda,\nu,\theta)&=&  (\nu+\theta)\Big(1-(\nu-\theta)^2\Big) -(\lambda+1)
\Big((1-\nu-\theta)(1+2\lambda+\nu+\theta)+4 \nu \theta\Big)\\
V(\lambda,\nu,\theta)&=&
(1-\nu-\theta)(1-\nu+\theta)(1+\nu-\theta)\\
&&\qquad \times (1+\nu-\theta+2\lambda)(1-\nu+\theta+2\lambda)(1-\nu-\theta+2\lambda)
\end{eqnarray*}

From the formula for the resolvent we can compute the density
$z\in [a,b]$:
\begin{eqnarray}
&&\rho(z)={\rm Im}\, F(z+i0)/\pi=1+{2\over \pi} \left(  {\rm Arctan}\, \sqrt{b-z\over z-a}\sqrt{a-\nu\over b-\nu}
- {\rm Arctan}\, \sqrt{b-z\over z-a}\sqrt{a\over b} \right. \label{densitwotri} \\
&&\qquad\qquad \left.
+{\rm Arctan}\, \sqrt{z-a\over b-z}\sqrt{\theta-b\over \theta-a}
- {\rm Arctan}\, \sqrt{z-a\over b-z}\sqrt{\lambda+1-b\over \lambda+1-a}\, \right)\nonumber \\
\end{eqnarray}
valid for $\nu,\theta$ satisfying $\nu<a<b<\theta$. Note that $\rho(a)=\rho(b)=1$.
\end{proof}

\subsubsection{}
When the size of any of the forbidden intervals in sufficiently small, it merges with the nearest end of the interval $[a,b]$. For example when
$$
\theta=\theta_c={1\over 3}\Big(1+\lambda+\nu+\sqrt{3(1+2\lambda)+(1+\lambda-2\nu)^2}\Big)
$$
the lower forbidden interval connects to $[a,b]$: $b=\theta_c$. When $\theta<\theta_c$ they remain connected, i.e. in this case we still have $b=\theta$.

The same holds for the upper forbidden interval, in which case when
$$
\nu=\nu_c={1\over 3}\Big(1+\lambda+\theta-\sqrt{3(1+2\lambda)+(1+\lambda-2\theta)^2}\Big)
$$
the upper forbidden interval connects with $[a,b]$, $a=\nu_c$
and they stay connected for $\nu <\nu_c$.

The assumption $\nu<a<b<\theta$ which we used in the previous section holds when $\theta <\theta_c$ and $\nu>\nu_c$.

In $\theta \geq\theta_c$ and $\nu > \nu_c$ the region where
$\rho(u)$ is not constant becomes $[a,\theta_c]$ with
\begin{equation*}
a={b (1+\lambda+\nu-b)^2\over \lambda^2}
\end{equation*}
The density in this case can be obtained from \eqref{2FI-dens} by
setting $b, \theta\to \theta_c$. The result is
\begin{multline}\label{den-2}
\rho(z)={2\over \pi} \left(  {\rm Arctan}\, \sqrt{\theta_c-z\over z-a}\sqrt{a-\nu\over \theta_c-\nu}
- \right. \\ \left.{\rm Arctan}\, \sqrt{\theta_c-z\over z-a}\sqrt{a\over \theta_c}- {\rm Arctan}\,
\sqrt{z-a\over \theta_c-z}\sqrt{\lambda+1-\theta_c\over \lambda+1-a}\, \right)
\end{multline}
It is easy to see that now the density satisfies $\rho(a)=1,\rho(b)=0$.

The formula above is valid as long as $\nu >\nu'_c$, namely
where
$$
\nu_c'=\nu_c|_{\theta=\theta_c}={1+\lambda-\sqrt{1+2\lambda}\over 2}
$$
When $\nu\leq \nu'_c$ the upper forbidden interval connects to
the interval $[a,b]$. In this case
\[
a={1+\lambda-\sqrt{1+2\lambda}\over 2}, \ \
b={1+\lambda+\sqrt{1+2\lambda}\over 2}
\]
The density inside of $[a,b]$ can be obtained from \eqref{den-2}
by taking $a, \nu\to\nu'_c$. The result is
\[
\rho(z)=\sqrt{{(\theta'_c-z)(z-\nu'_c)\over \theta'_c\nu'_c}}
\]
which agrees with \eqref{uni-densit}.

\begin{remark}The transition from free to trapped paths is similar to that found by Douglas and Kazakov in QCD.
\end{remark}

\subsubsection{} When $\theta+\nu=1$ the interval $[a,b]$ collapses into a point.

\section{Case of the full hexagon}\label{hexagon}

\begin{figure}
\centering
\includegraphics[width=14.cm]{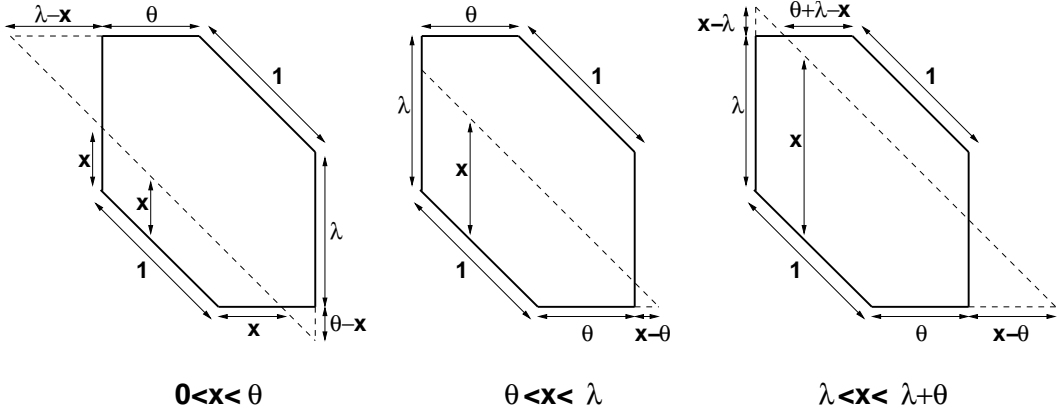}
\caption{We decompose the hexagon $1\times\theta\times\lambda$
into two cut hexagons with one side equal to $1$, and with possible
triangular holes at corners. Three cases must be considered.}\label{fig:hexa}
\end{figure}

In this section we will reproduce the results of \cite{CLP}
for the limit shape of the hexagon tiling and will prepare the ground for the next section where we will use the symmetries of
regions to find the limit density at free boundaries for more complicated regions then in previous sections.

Cut the hexagon as it is illustrated in Fig.\ref{fig:hexa}
into two cut hexagons. In this section we will study the
limit density of tiles by studying the
the limit density of lattice path points
intersecting an arbitrary cut (dashed line) which is parallel to one side of the hexagon . We assume that the hexagon has size $k\times n\times p$, and that the cut is at a distance $m$ from the edge of length $k$. Our goal to pass to
the limit $k\to \infty$ with $n/k=\lambda$, $p/k=\theta$, $m/k=x$
and to compute the limit density of lattice paths through the cut.
We will assume $\lambda\geq \theta$ which clearly does not restrict generality.

The cut separates the hexagon into two regions.
For the purpose of counting lattice paths intersecting
the cut we distinguish three possible cases of geometries
of these regions:
\begin{itemize}
\item the lower region is a cut hexagon, while the
upper region is a cut hexagon with two forbidden intervals,
\item both lower and upper regions have one forbidden interval,
\item the lower region is a cut hexagon with two forbidden intervals, while the upper region is a cut hexagon.
\end{itemize}

All three cases are illustrated on Fig.\ref{fig:hexa} and
have been treated in the previous sections. The new ingredient is that
we must now glue the two halves, by identifying the endpoints $m_i \leftrightarrow m_{k+1-i}'$
of both families of paths.

In the large $k$ limit the density of paths through the cut develops the limit shape which minimizes the action functional determined by the asymptotic of the counting formula or the lattice paths.
The formula is the same for both pieces of the hexagon.
The resulting action functional is the sum of two pieces.

To describe the resulting action let us parameterize the endpoints of paths of the lower half of the hexagon by $\mu_i=m_i/k$  .
In all three cases of Fig.\ref{fig:hexa}, the range of $\mu_i$ is $[0,1+x]$. In the second case there is an extra
restriction $\mu_i<1+\theta$ corresponding to the forbidden interval $[1+\theta,1+x]$, and in the third case the restriction $x-\lambda<\mu_i<1+\theta$, corresponds to two forbidden intervals
$[0, x-\lambda]\cup [1+\theta, 1+x]$.

For the upper cut hexagon the range of endpoints $\mu_i'=m_i'/k$
of lattice paths originated on the upper left side is $[0,1+\lambda+\theta-x]$ in all three cases of Fig.\ref{fig:hexa}, while restrictions apply in the first case $\theta-x<\mu_i'<1+\theta$,
and in the second case $\mu_i'<1+\theta$.
The upper and lower half correspond to an admissible rhombus tiling if and only if the endpoints $\mu_i$
and $\mu_{k+1-i}'$ coincide. It gives the identity
$$
\mu_{k+1-i}'=1 +\theta-\mu_i
$$
which holds in all three cases.

In the limit $k\to \infty$ the asymptotic of the counting formula
of lattice paths gives the action functional
\[
S[\rho]= S_x[\rho]+  S_{\lambda+\theta-x}[\rho^\theta]
\]
where $\rho^\theta(\mu)=\rho(\theta-\mu)$ and
$S_\lambda[\rho]$ is the action \eqref{action-uni}.

The limit density is the minimizer of this
functional with the restrictions described above.
The extrema of $S[\rho]$ are solutions to the integral equation
\begin{equation}\label{simple}
{1\over 2}\,{\rm Log}\, {\mu(\lambda-x+\mu) \over (1+x-\mu)(1+\theta-\mu)} =\dashint_I {\rho(u) du\over \mu-u}
\end{equation}
subject to the restriction $\int_I \rho(u)du=1$.
Here $I=[{\rm Max}(0,x-\lambda), 1+{\rm Min}(x,\theta)]$.

\begin{thm} Depending on the geometry of the cut hexagon, the limit density has the following form.
\begin{enumerate}
\item The limit density $\rho(z)$ is supported on the interval $[a,b]\subset [{\rm Max}(0,x-\lambda), 1+{\rm Min}(x,\theta)]$ with
\begin{eqnarray}
a&=&\left({\sqrt{\lambda(\lambda+\theta-x)}-\sqrt{x\theta(1+\lambda+\theta)}
\over \lambda+\theta}\right)^2 \nonumber \\
b&=&\left({\sqrt{\lambda(\lambda+\theta-x)}+\sqrt{x\theta(1+\lambda+\theta)}
\over \lambda+\theta}\right)^2 \label{abhexa}
\end{eqnarray}
and on this interval is given by
\begin{eqnarray}
\rho(z)&=&{1\over \pi}\left({\rm Arctan}\, \sqrt{b-z\over z-a}\sqrt{1+\theta-a\over 1+\theta-b}
+{\rm Arctan}\,  \sqrt{b-z\over z-a}\sqrt{1+x-a\over 1+x-b} \right. \nonumber \\
&&\qquad\qquad \left. - {\rm Arctan}\,  \sqrt{b-z\over z-a}\sqrt{a\over b}
-{\rm Arctan}\,  \sqrt{b-z\over z-a}\sqrt{\lambda-x+a\over \lambda-x+b}
\right)\label{densitone}
\end{eqnarray}
Outside of $[a,b]$ the limit density is zero.

\item The limit density $\rho(z)$ is $1$ outside $[a,b]$ i.e.
on $[{\rm Max}(0,x-\lambda),a]\cup [b,1+{\rm Min}(x,\theta)]]$,  while inside $[a,b]$ it is given by
\begin{eqnarray}
\rho(z)&=&1+{1\over \pi}\left( {\rm Arctan}\, \sqrt{b-z\over z-a}\sqrt{1+\theta-a\over 1+\theta-b}
-{\rm Arctan}\,  \sqrt{b-z\over z-a}\sqrt{1+x-a\over 1+x-b} \right. \nonumber \\
&&\qquad\qquad \left. +{\rm Arctan}\,  \sqrt{b-z\over z-a}\sqrt{a\over b}
-{\rm Arctan}\,  \sqrt{b-z\over z-a}\sqrt{\lambda-x+a\over \lambda-x+b}
\right)\label{densitwo}
\end{eqnarray}

\item $\rho(z)=1, z\in [{\rm Max}(0,x-\lambda),a]$, and $\rho(z)=0, z\in [b, 1+{\rm Min}(x,\theta)]]$ .
In this case $a$ and $b$ are given by the same formulae as above while
\begin{eqnarray}
\rho(z)&=&{1\over \pi}\left( {\rm Arctan}\, \sqrt{b-z\over z-a}\sqrt{1+\theta-a\over 1+\theta-b}
+{\rm Arctan}\,  \sqrt{b-z\over z-a}\sqrt{1+x-a\over 1+x-b} \right.\nonumber  \\
&&\qquad\qquad \left. +{\rm Arctan}\,  \sqrt{b-z\over z-a}\sqrt{a\over b}
-{\rm Arctan}\,  \sqrt{b-z\over z-a}\sqrt{\lambda-x+a\over \lambda-x+b}
\right)\label{densitri}
\end{eqnarray}

\item $\rho(z)=1, z\in [{\rm Max}(0,x-\lambda),a]$, $\rho(z)=0, z\in [b, 1+{\rm Min}(x,\theta) ]$. In this case
\begin{eqnarray}
\rho(z)&=&1+{1\over \pi}\left( {\rm Arctan}\, \sqrt{b-z\over z-a}\sqrt{1+\theta-a\over 1+\theta-b}
-{\rm Arctan}\,  \sqrt{b-z\over z-a}\sqrt{1+x-a\over 1+x-b} \right. \nonumber \\
&&\qquad\qquad \left. -{\rm Arctan}\,  \sqrt{b-z\over z-a}\sqrt{a\over b}
-{\rm Arctan}\,  \sqrt{b-z\over z-a}\sqrt{\lambda-x+a\over \lambda-x+b}
\right)\label{densifor}
\end{eqnarray}

\end{enumerate}
\end{thm}
\begin{proof}

{\bf 1.} Let us first assume that the density $\rho$ is supported on an interval $[a,b]$ with ${\rm Max}(0,x-\lambda)<a<b<1+{\rm Min}(x,\theta)$, and that $\rho(a)=\rho(b)=0$. We must solve
eq.\eqref{simple} with the normalization $\int_a^b\rho(u)du=1$. The corresponding resolvent reads:
\begin{eqnarray*}
&&F(z)={\sqrt{w(w-1)}\over 2\pi}\int_0^1 {dv\over \sqrt{v(1-v)}} \, {1\over w-v}\,{\rm Log}\,
{(\beta+v)(\delta+v) \over (\gamma-v)(\eta-v)} \\
&&={1\over 2}\, {\rm Log}\, {(\beta+w)(\delta+w) \over (\gamma-w)(\eta-w)}
+i\left( {\rm Arctan}\, \sqrt{1-w\over w}\sqrt{\eta\over \eta-1}
+{\rm Arctan}\, \sqrt{1-w\over w}\sqrt{\gamma\over \gamma-1}
\right.\\
&&\qquad \qquad \qquad \left. - {\rm Arctan}\, \sqrt{1-w\over w}\sqrt{\beta\over 1+\beta}
-{\rm Arctan}\, \sqrt{1-w\over w}\sqrt{\delta\over 1+\delta}
\right)
\end{eqnarray*}
where $w={z-a\over b-a}$ and:
$$
\beta={a\over b-a}, \quad \gamma={1+x-a\over b-a}, \quad \delta={\lambda-x+a\over b-a}, \quad \eta={1+\theta-a\over b-a}.
$$
The constrain on the density translates to the condition $F(z)= 1/z+O(1/z^2)$, as $z\to \infty$ on the resolvent. Using
identities \eqref{firstint}\eqref{secondint}
we obtain the equations:
\begin{eqnarray*}
&&\Big(\sqrt{\beta+1}+\sqrt{\beta}\Big)\Big(\sqrt{\delta+1}+\sqrt{\delta}\Big)=
\Big(\sqrt{\gamma}+\sqrt{\gamma-1}\Big)\Big(\sqrt{\eta}+\sqrt{\eta-1}\Big)\\
&&1={b-a\over 4}\left( (\sqrt{\beta+1}-\sqrt{\beta})^2+(\sqrt{\gamma}-\sqrt{\gamma-1})^2
+(\sqrt{\delta+1}-\sqrt{\delta})^2+(\sqrt{\eta}-\sqrt{\eta-1})^2 \right)
\end{eqnarray*}
Solving these equations for $a$ and $b$ we get the formulae
\eqref{abhexa}.

{\bf 2.} Next, let us assume the density is one near the boundaries, and varies only on an interval $[a,b]$.
This happens, for example, for sufficiently small $x$,
when $0<x<\theta$ (see the left figure on Fig.\ref{fig:hexa}).  In this case $\rho(z)=1$ for $z\in[0,a]\cup [b,1+x]$.
Eqn.\eqref{simple} turns into:
\begin{equation}\label{double}
{1\over 2}\,{\rm Log}\, {(\mu-a)^2(1+x-\mu)(\lambda-x+\mu) \over \mu(b-\mu)^2(1+\theta-\mu)}
=\dashint_a^b {\rho(u) du\over \mu-u}
\end{equation}
and we have
\begin{eqnarray*}
&&F(z)={\sqrt{w(w-1)}\over 2\pi}\int_0^1 {dv\over \sqrt{v(1-v)}}\,{1\over w-v}\,{\rm Log}\,
{v^2(\gamma-v)(\delta+v) \over(1-v)^2(\beta+v)(\eta-v)}\\
&&={1\over 2}\,{\rm Log}\, {w^2(\gamma-w)(\delta+w) \over(1-w)^2(\beta+w)(\eta-w)}
+i\left( {\rm Arctan}\, \sqrt{1-w\over w}\sqrt{\eta\over \eta-1}
-{\rm Arctan}\, \sqrt{1-w\over w}\sqrt{\gamma\over \gamma-1}
\right.\\
&&\qquad \qquad \qquad \left. +{\rm Arctan}\, \sqrt{1-w\over w}\sqrt{\beta\over 1+\beta}
-{\rm Arctan}\, \sqrt{1-w\over w}\sqrt{\delta\over 1+\delta}
\right)
\end{eqnarray*}
where
$$
\beta={a\over b-a}, \quad \gamma={1+x-a\over b-a},\quad \delta={\lambda-x+a\over b-a},\quad
\eta={1+\theta-a\over b-a}
$$
The constraint for the density implies $F(z)= {b-a-x\over z}+O({1\over z^2}$ for large $z$. This gives:
\begin{eqnarray*}
&& \Big(\sqrt{\beta+1}+\sqrt{\beta}\Big)\Big(\sqrt{\eta}+\sqrt{\eta-1}\Big)=
\Big(\sqrt{\delta+1}+\sqrt{\delta}\Big)\Big(\sqrt{\gamma}+\sqrt{\gamma-1}\Big) \\
&& b-a-x={b-a\over 4}\left(4-(\sqrt{\beta+1}-\sqrt{\beta})^2-(\sqrt{\gamma}-\sqrt{\gamma-1})^2 \right.\\
&&\qquad\qquad \qquad\qquad  \qquad\left.
+(\sqrt{\delta+1}-\sqrt{\delta})^2+(\sqrt{\eta}-\sqrt{\eta-1})^2\right)
\end{eqnarray*}
Remarkably, this leads to the same result \eqref{abhexa} as above, and for the density we obtain \eqref{densitwo}

{\bf 3.} In mixed cases where $\rho(z)=1$ for $z<a$ and $\rho(z)=0$ for $z>b$  the resolvents are:
\begin{eqnarray*}
F(z)&=&{\sqrt{w(w-1)}\over 2\pi}\int_0^1 {dv\over \sqrt{v(1-v)}}\,{1\over w-v}\,{\rm Log}\,
{v^2(\delta+v) \over (\beta+v)(\gamma-v)(\eta-v)}\\
F(z)&=&{\sqrt{w(w-1)}\over 2\pi}\int_0^1 {dv\over \sqrt{v(1-v)}}\,{1\over w-v}\,{\rm Log}\,
{(\beta+v)(\gamma-v)(\delta+v) \over(1-v)^2(\eta-v)}
\end{eqnarray*}
These formulae lead to the same formulae for $a$ and $b$ and to the formulae
\eqref{densitri} \eqref{densifor} for the densities.

\end{proof}

\begin{remark}
At the boundary intervals of length $1$ (they correspond to
cuts at $x=0$ and $x=\lambda+\theta $) the limit density is
 $1$. This can be derived from the formulae in the theorem
 and from the identities
\begin{eqnarray}
\sqrt{\gamma}+\sqrt{\gamma-1}&=&
{\rm Max}(x \lambda(1+\lambda+\theta),\theta(\lambda+\theta-x))^{1\over 4}\nonumber \\
\sqrt{\gamma}-\sqrt{\gamma-1}&=&
{\rm Min}(x \lambda(1+\lambda+\theta),\theta(\lambda+\theta-x))^{1\over 4}\nonumber \\
\sqrt{\beta+1}+\sqrt{\beta}&=&
{\rm Max}(\lambda(\lambda+\theta-x),x\theta(1+\lambda+\theta))^{1\over 4}\nonumber \\
\sqrt{\beta+1}-\sqrt{\beta}&=&
{\rm Min}(\lambda(\lambda+\theta-x),x\theta(1+\lambda+\theta))^{1\over 4}\nonumber \\
\sqrt{\eta}+\sqrt{\eta-1}&=&
{\rm Max}(x \lambda,\theta(1+\lambda+\theta)(\lambda+\theta-x))^{1\over 4}\nonumber \\
\sqrt{\eta}-\sqrt{\eta-1}&=&
{\rm Min}(x \lambda,\theta(1+\lambda+\theta)(\lambda+\theta-x))^{1\over 4}\nonumber \\
\sqrt{\delta+1}+\sqrt{\delta}&=&
{\rm Max}(x\theta,\lambda(\lambda+\theta-x)(1+\lambda+\theta))^{1\over 4}\nonumber \\
\sqrt{\delta+1}-\sqrt{\delta}&=&
{\rm Min}(x\theta,\lambda(\lambda+\theta-x)(1+\lambda+\theta))^{1\over 4}\label{minmax}
\end{eqnarray}
which, in turn, implies
\begin{eqnarray}
1+\theta-a&=&
\left({\sqrt{ x \lambda}+\sqrt{\theta(1+\lambda+\theta)(\lambda+\theta-x)}
\over \lambda+\theta} \right)^2 \nonumber \\
1+\theta-b&=&
\left({\sqrt{ x \lambda}-\sqrt{\theta(1+\lambda+\theta)(\lambda+\theta-x)}
\over \lambda+\theta}\right)^2 \nonumber \\
1+x-a&=&
\left({\sqrt{ x \lambda(1+\lambda+\theta)}+\sqrt{\theta(\lambda+\theta-x)}
\over \lambda+\theta} \right)^2 \nonumber \\
1+x-b&=&
\left({\sqrt{ x \lambda(1+\lambda+\theta)}-\sqrt{\theta(\lambda+\theta-x)}
\over \lambda+\theta} \right)^2 \nonumber \\
\lambda-x+a&=&
\left({\sqrt{ x\theta}-\sqrt{\lambda(\lambda+\theta-x)(1+\lambda+\theta)}
\over \lambda+\theta} \right)^2 \nonumber \\
\lambda-x+b&=&
\left({\sqrt{ x\theta}+\sqrt{\lambda(\lambda+\theta-x)(1+\lambda+\theta)}
\over \lambda+\theta} \right)^2 \label{xes}
\end{eqnarray}
\end{remark}

\begin{figure}
\centering
\includegraphics[width=8.cm]{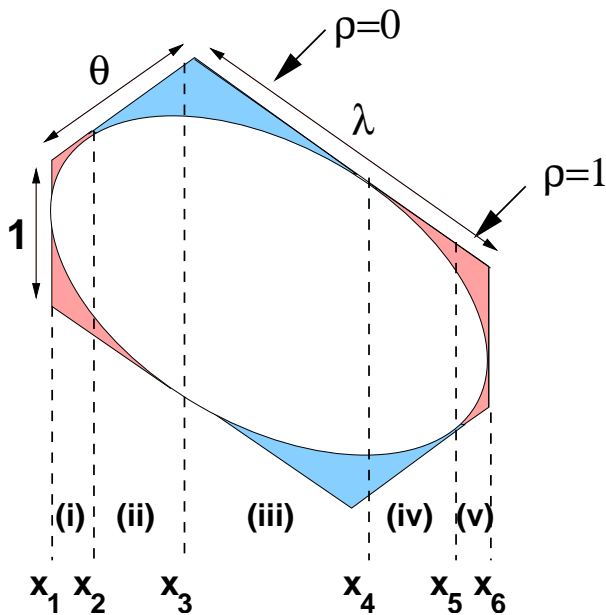}
\caption{The full solution for the density $\rho$ of path points (by vertical sections)
in the case of an hexagon of size $1\times \lambda\times\theta$.
The pink regions correspond to $\rho=1$, while the blue ones stand for $\rho=0$, both
frozen.
The domain of segments $[a,b]$ over which $\rho$ varies sweeps the inner ellipse.}\label{fig:hexafreeze}
\end{figure}

The general picture is summarized in Fig.\ref{fig:hexafreeze}.
The density $\rho$ varies only within the ellipse inscribed inside the hexagon. The points where the ellipse touching
the boundary of the hexagon have $x$-coordinates
$x_1=0$ and $x_6=\lambda+\theta$ (vertical tangents), $x_2={\theta\over 1+\lambda}$,
$x_3={\lambda\over 1+\theta}$, $x_4={\theta(1+\lambda+\theta)\over 1+\theta}$,
and $x_5={\lambda(1+\lambda+\theta)\over 1+\lambda}$,
with $x_1\leq x_2\leq x_3\leq x_4\leq x_5\leq x_6$.
They correspond, respectively, to
$b=a$, $b=1+x$, $a=0$, $b=1+\theta$, and $a=x-\lambda$.

The limit density is constant outside of the ellipse   Fig.\ref{fig:hexafreeze} and is given by:

{(i)} $x\in[x_1,x_2]$: top and bottom $\rho =1$

{(ii)} $x\in [x_2,x_3]$: top $\rho =0$, bottom $\rho =1$

{(iii)} $x\in [x_3,x_4]$: top and bottom $\rho =0$

{(iv)} $x\in [x_4,x_5]$: top $\rho =1$, bottom $\rho =0$

{(v)} $x\in [x_5,x_6]$: top and bottom $\rho =1$.

The explicit expressions for $\rho$ inside $[a,b]$ are given by eq.\eqref{densitwo} in the cases (i) and (v),  eq.\eqref{densitri} in the case (ii), eq.\eqref{densitone}
in the case (iii), and eq.\eqref{densifor} in the case (iv).

\section{Comparison with the general results of Kenyon-Okounkov-Sheffield}\label{burgers}

\subsection{Equations for the slope}
\subsubsection{}
Limit shapes for rhombi tilings and more generally for
dimer models on bipartite graphs are known to be solutions
to certain variational problem \cite{CKP}\cite{KOS}. They minimize the rate functional which is determined by the free energy of the system on a torus.

Assume we have a polygonal region with sides $s_1,\dots, s_{3d}$ of lengths $N_1,\dots, N_{3d}$ parallel to
the axes of the triangular lattice, such that three possible
axes of the lattice alternate in the sequence $s_i$.
Consider random tilings of such region with the probability
$\propto q^{vol(h)}$ where $h$ is the height function corresponding to the tiling. When $q\to 1$ such that $q^{-\epsilon c}$ and
$\epsilon N_i$ are kept finite random tilings develop the limit shape
at the macroscopical scale \cite{KO}. Moreover in
this case the limit shape is described by an algebraic curve which can be explicitly constructed from the geometric data of the region.

Recall that rhombus tilings of a simply connected domain are in
a bijection with the height functions (certain discrete surfaces in 3D over this domain).
When the domain is not simply connected the height function does not exist but
what would be its gradient still makes sense.
It is called the {\bf slope} of the tiling.
The slope describes the local
proportions of tiles of types $a,b,c$ as follows. Consider the tiling of a toroidal domain of fixed size,
let $N_a,N_b,N_c$ be the numbers of tiles used in the tiling with $N=N_a+N_b+N_c$, then the slope is defined as
$(s,t)=\nabla h=(h_x,h_y)=(N_a/N,N_b/N)$ (here, $x,y$ refer to coordinates in a basis $(u,v)$, with
$u^2=v^2=1$ and $u.v =-1$).

\subsubsection{} As explained in \cite{KOS}\cite{KO} there
are two natural complex functions on the domain in question,
$z$ and $w$ parameterizing the height function:
\begin{equation}\label{slop}
 \nabla h={1\over \pi}\left({\rm Arg}(-w),{\rm Arg}(-{1\over z})\right)
\end{equation}
In the limit when $q\to 1$ and the regions increase
at the consistent rate (see above),  these variables satisfy the system of equations:
\begin{eqnarray*}
P(z,w)&=&0\\
{\rm Re}(z_x/z+w_y/w)&=& c
\end{eqnarray*}
for some constant $c$.
In addition $z$ and $w$ should satisfy the boundary conditions
determined by the geometry of the region and restrictions on tilings.

The function $P(z,w)$ is polynomial for periodically weighted lattices. It is the determinant of the Kasteleyn matrix
with quasiperiodic\footnote{Quasiperidicity means there are extra factors $z^{\pm1}$ and $w^{\pm1}$ for the dimers
crossing two given lines parallel to the $x$ and $y$ directions
respectively, depending on the Kasteleyn orientation of the corresponding edges.}
boundary conditions which is the partition function of the dimer model on the(toroidal)
fundamental domain of its weights.

In our case
\[
P(z,w)=1+z+w .
\]

Solving the system for $z$ and $w$ by complex characteristics
we change it to the system of equations
\begin{eqnarray}\label{PQ-syst}
P(z,w)&=&0\\
Q(e^{-c x} z,e^{-c y} w)&=&0
\end{eqnarray}
The function $Q$ is an analytic function, which depends on the
boundary of the region.

When the region is polygonal, the function $Q$ is polynomial.
It can be computed according to the description of \cite{KO} (Section 3.3.3).

One of the remarkable phenomena we observe in tilings of large
domains is the existence of the frozen regions near the boundary.
In these regions the height function is linear. In the case
of rhombus tilings of a hexagon, the curve separating the
frozen part and the disordered phase (where the height function
varies smoothly) is an ellipse \cite{CLP}, and it was called the arctic circle.
In \cite{KO} such curves are called clout curves.
One of the important corollaries is that for polygonal regions
such curves are real algebraic in $e^{cx}, e^{cy}$. It is also
important that they form a family of curves as one varies $c$.

When $c\to 0$, the second equation in \eqref{PQ-syst} becomes
 $Q(z,xz+yw)=0$.

The arctic circle (the boundary of the limit shape) in
the case of polygonal region is an algebraic curve
intersecting each linear piece of the boundary once
where it is tangent to the boundary. It may develop
cusps. It may also be disconnected, i.e. it may have
"islands" and bubbles" (in the terminology of \cite{KO}),
see Fig. 5, fig. 17
from \cite{KO}.

\subsection{Case of the cut hexagon}

As explained before, the limiting density of exit points of LGV paths on a large cut hexagon $k\times n\times n$
with $n/k=\lambda$ fixed, $k\to\infty$ is identical to that of passing points of LGV paths through the main diagonal (which is where we make the cut). These paths describe
the rhombus tilings of the hexagon which is obtained from
the cut hexagon by gluing to the reflection w.r.t. the cut. In the language of \cite{KO},
this density is the density of tilted rhombi and is determined by
the limit shape of rhombus tilings of the complete hexagon.

For the unweighted tilings by rhombi we have $P(z,w)=1+z+w$.
The polynomial $Q$ for such tiling of the hexagon has degree $2$ ( half of the number of sites of the region) and reads:
$$
Q(z,x z+y w)=\lambda^2 (z+z^2)-(xz+y w)^2-(1+\lambda)(xz+yw)
$$
The equations \eqref{PQ-syst} for $c=0$ become quadratic equation for $z$:
$$
(\lambda^2-(x-y)^2)z^2+(\lambda^2-(1+\lambda-2y)(x-y))z +y(1+\lambda-y)=0
$$
The arctic circle is the curve in $(x,y)$ on which the discriminant of this equation vanishes:
\begin{eqnarray*}
0&=&(\lambda^2-(1+\lambda-2y)(x-y))^2-4y(1+\lambda-y)(\lambda^2-(x-y)^2)\\
&=& (1+2\lambda)(x-y)^2+\lambda^2(x+y-\lambda-1)^2-\lambda^2(1+2\lambda)
\end{eqnarray*}
It is tangent to the six edges of the hexagon, namely $y=0,y=\lambda+1$,
$x=0,x=\lambda+1$, and $y=x\pm \lambda$.
This curve intersects the reflection axis $y=x$ (the cut) at points
$x={\lambda+1\over 2}\pm {\sqrt{1+2\lambda}\over 2}$
which coincide with the ends of the support of $\rho$ .

The passing points of the LGV paths along the cut correspond to tiles of type $a$ or $b$.
The local density of such tiles is therefore $(N_a+N_b)/N=h_x+h_y$. Hence the density
\eqref{densitone} should be compared to the ``vertical" slope $h_x+h_y$ in our solution.

Solving for $z$ at $y=x$, we find that
$\lambda^2(z+{1\over 2})^2=(x-{1+\lambda\over 2})^2-{1+2\lambda\over 4}$, so that:
\begin{eqnarray*}
z&=&{1\over 2}\left(-1\pm \sqrt{(x-a)(x-b)\over a b}\right)\\
w&=&{1\over 2}\left(-1\mp \sqrt{(x-a)(x-b)\over a b}\right)
\end{eqnarray*}
and finally we get
$$
h_x=h_y={1\over \pi} {\rm Arctan}\,  \sqrt{(b-x)(x-a)\over a b}
$$
and therefore $\rho(x)=h_x+h_y$ as expected.

\subsection{Case of the full hexagon}

Here we will compare the results of previous sections with the
computation of the limit shape based on the equation \eqref{PQ-syst}.

In the case of the full hexagon with edges $1,\lambda,\theta$ the
equations \eqref{PQ-syst} becomes
\begin{eqnarray*}
0&=&1+z+w\\
0&=&\lambda \theta (z+z^2) -(xz+yw)^2-(\lambda-\theta)z(xz+yw)-(1+\lambda)(xz+yw)
\end{eqnarray*}
The arctic curve appear as the condition that the
discriminant of this equation vanishes. It is the ellipse:
\begin{equation}\label{arctichexa}
{\Big((1+\lambda)x+(1+\theta)y-(1+\lambda)(1+\theta)\Big)^2\over 1+\lambda+\theta}+
{\Big((1+\lambda)x-(1+\theta)y\Big)^2\over \lambda\theta}=(1+\lambda)(1+\theta)
\end{equation}
tangent to respectively $x=0$, $x=\theta$, $y=0$, $y=\lambda$, $y=x+\lambda$ and
$y=x-\theta$.

To get the density of LGV points, let us
solve for $z$ on some arbitrary vertical line $y=x+t$.
It is easy to see that the line crosses the arctic curve
at the points
$$
x_\pm=\left({\sqrt{\lambda(t+\theta)}\mp \sqrt{\theta(\lambda-t)(1+\lambda+\theta)}
\over \lambda+\theta}\right)^2
$$
These are identified with the ends of the support of $\rho$ given in eq.\eqref{abhexa},
namely $x_-=a$ and $x_+=b$,
upon the substitution $x=\lambda-t$.
The solution for $z$ reads:
$$
z=-{1\over 2(t+\theta)(\lambda-t)}\left(t+(\lambda-t-x)(t+\theta)+(\lambda-t)(t+x)\pm(\lambda+\theta)\sqrt{(x-a)(x-b)}\right)
$$
and we finally get
\begin{eqnarray*}
h_x&=&{1\over \pi}{\rm Arctan}\, {(\lambda+\theta)\sqrt{(b-x)(x-a)}\over
(\lambda-t-x)(t+\theta)+t+(\lambda-t)(t+x)}\\
h_y&=&{1\over \pi}{\rm Arctan}\, {(\lambda+\theta)\sqrt{(b-x)(x-a)}\over
(\lambda-t+x)(t+\theta)-t-(\lambda-t)(t+x)}
\end{eqnarray*}
Using the addition formula for Arctan, we finally get
\begin{equation}\label{slopsum}
h_x+h_y={1\over \pi}{\rm Arctan}\, {(\lambda+\theta)\sqrt{(b-x)(x-a)}\over
\lambda\theta-t(1+\theta)+2x(x+t-1-{\lambda+\theta\over 2}) }
\end{equation}

To compare this with the results on the LGV point densities, let us assume for definiteness
that we have $x_3<x<x_4$ (case (iii)), namely
${\lambda\over 1+\theta}-\theta<t<{\lambda\theta\over 1+\theta}$.
Using eqs.
\eqref{abhexa} and \eqref{xes} and the restriction on $t$, we find that
\begin{eqnarray*}
\sqrt{1+\theta-a\over 1+\theta-b}&=&
{\sqrt{\theta(\theta+t)(1+\lambda+\theta)}+\sqrt{\lambda(\lambda-t)}\over
\sqrt{\theta(\theta+t)(1+\lambda+\theta)}-\sqrt{\lambda(\lambda-t)}}\\
\sqrt{a\over b}&=&
{\sqrt{\theta(\lambda-t)(1+\lambda+\theta)}-\sqrt{\lambda(\theta+t)}\over
\sqrt{\theta(\lambda-t)(1+\lambda+\theta)}+\sqrt{\theta(\theta+t)}}\\
\sqrt{1+x-a\over 1+x-b}&=&
{\sqrt{ \lambda(\lambda-t)(1+\lambda+\theta)}+\sqrt{\theta(\theta+t)}\over
\sqrt{ \lambda(\lambda-t)(1+\lambda+\theta)}-\sqrt{\theta(\theta+t)}}\\
\sqrt{\lambda-x+a\over \lambda-x+b}&=&
{\sqrt{\lambda(\theta+t)(1+\lambda+\theta)}-\sqrt{\theta (\lambda-t)}\over
\sqrt{\lambda(\theta+t)(1+\lambda+\theta)}+\sqrt{\theta (\lambda-t)}}
\end{eqnarray*}
We use the subtraction/addition formulae for Arctan to get:
\begin{eqnarray*}
&&{\rm Arctan}\, \sqrt{b-z\over z-a}\sqrt{1+x-a\over 1+x-b}-{\rm Arctan}\, \sqrt{b-z\over z-a}\sqrt{a\over b}=
{\rm Arctan}\, {\sqrt{(b-z)(z-a)}(\lambda+\theta)\over \theta\lambda-t(1+\theta)+z(\lambda-\theta)}\\
&&{\rm Arctan}\, \sqrt{b-z\over z-a}\sqrt{1+\theta-a\over 1+\theta-b}-
{\rm Arctan}\, \sqrt{b-z\over z-a}\sqrt{\lambda-x+a\over \lambda-x+b}=
{\rm Arctan}\, {\sqrt{(b-z)(z-a)}(\lambda+\theta)\over \theta\lambda+t(1+\theta)-z(\lambda-\theta)}
\end{eqnarray*}
and finally
\begin{eqnarray*}
\rho(z)&=&{1\over \pi}\left({\rm Arctan}\, {\sqrt{(b-z)(z-a)}(\lambda+\theta)
\over \theta\lambda-t(1+\theta)+z(\lambda-\theta)}
+{\rm Arctan}\, {\sqrt{(b-z)(z-a)}(\lambda+\theta)
\over  \theta\lambda+t(1+\theta)-z(\lambda-\theta)}\right)\\
&=&{1\over \pi} {\rm Arctan}\, {\sqrt{(b-z)(z-a)}(\lambda+\theta)\over
\theta\lambda-t(1+\theta)+2z(z+t-1-{\lambda+\theta\over 2}) }
\end{eqnarray*}
in agreement with \eqref{slopsum}.

The other cases follow analogously.

\section{Symmetries}\label{symmetries}

\begin{figure}
\centering
\includegraphics[width=8.cm]{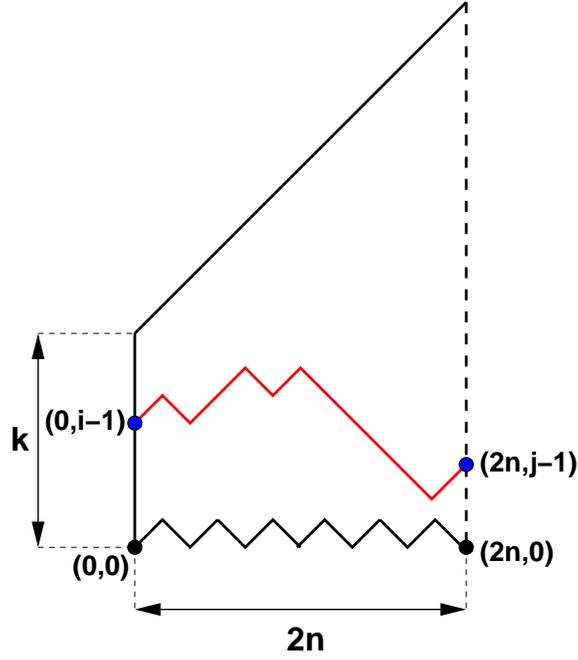}
\caption{A cut hexagon with zig-zag fixed boundary.}\label{fig:halfcut}
\end{figure}

\subsection{Case of a cut hexagon with a fixed zig-zag boundary}

\subsubsection{} The ``half" cut hexagon (see Fig.\ref{fig:halfcut}), of the
size $k\times 2n\times 2n+k$ is another case when the total number of lattice path configurations nicely factorizes. In this case the lattice paths start at each point of the left boundary and ends
on the right boundary. The low boundary is a horizontal zig-zag line.

The LGV formula expresses the total number of configurations of  such non-intersecting paths, which start
at $(0,2i-2)$, $i=1,2,...,k$, and end at $(2n,2m_j-2)$, $j=1,2,...,k$, $1\leq m_j\leq k+n$ as the determinant:
det$_{1\leq i,j\leq k}(p_{i,m_j}^{(n)})$. Here  $p_{i,m}^{(n)}$ is the numbers of paths from $(0,2i-2)$ to $(2n,2m-2)$.
By the standard reflection principle, this number $p_{i,m}^{(n)}$ is exactly the number of unrestricted paths
from $(0,2i-2)$ to $(2n,2m-2)$ minus the number of of unrestricted paths from $(0,-2i)$ to $(2n,2m-2)$:
$$
p_{i,m}^{(n)}={2n\choose n+m-i}-{2n\choose n+m+i-1}
$$
Finally we have the factorization formula
\begin{multline} \label{halfcuthex}
\tilde{Z}_{k,n}(\{\bm\})=\det_{1\leq i,j\leq k}(p_{i,m_j}^{(n)})=\prod_{i=1}^k {(2n+2i-2)!
\over (n+m_i+k-1)! (n-m_i+k)!}\\
\,\prod_{1\leq i<j\leq k}(m_j-m_i)
\,\prod_{1\leq i\leq j\leq k} (m_i+m_j-1)
\end{multline}

Let us now compute the asymptotic of this number, when $k\to \infty$, and say $n/k=\al$ fixed,
while $\mu_i=m_i/k$ converges to a continuous function on
$[0,1]$, or equivalently gets distributed on $[0,\alpha+1]$
with the limit density $\rho(\mu)$.

We have
\begin{multline}
-{1 \over k} {\rm Log}(\tilde{Z}_{k,n}(\{\bm\}))=C+{1\over k}\sum_{i=1}^k ((\al+1+\mu_i)\, {\rm Log}(\al+1+\mu_i)+\\
(\al+1-\mu_i)\, {\rm Log}(\al+1-\mu_i)) -{1\over k^2}\sum_{i<j} {\rm Log} |\mu_j^2-\mu_i^2|
\end{multline}
for some constant $C$ depending only on $\al$.
As $k\to \infty$ this gives the action (the rate functional)
\begin{multline}
S[\rho]= C+\int_0^{\alpha+1}((\al+1+\mu)\, {\rm Log}(\al+1+\mu)+
(\al+1-\mu)\, {\rm Log}(\al+1-\mu))\rho(\mu)d\mu-\\
\int_0^{\alpha+1}
\int_0^{\alpha+1}{\rm Log} |\mu^2-\nu^2|\rho(\mu)\rho(\nu)d\mu d\nu
\end{multline}
Fixed points of this functional are solutions to the integral equation
$$
{\rm Log}\, {\al+1+\mu \over \al+1-\mu} =\dashint_0^{\alpha+1}
\left({1\over \mu -u}+{1\over \mu+u}\right) \sigma(u)du
$$
subject to the constraint $\int_0\sigma(u)du=1$. As before
assume the density is zero for $z\in [a,\alpha+1]$.
It is convenient to extent the problem to the symmetric interval
$[-a,a]$ assuming that $\sigma(z)$ is extended as a symmetric function $\sigma(z)=\sigma(-z)$. The integral equation becomes
\[
{\rm Log}\, {\al+1+\mu \over \al+1-\mu} =\dashint_{-a}^{a}
{1\over \mu -u} \sigma(u)du
\]
with the constraint $\int_{-a}^a\sigma(u)du=2$.

Similarly to previous sections it is easy to obtain
an explicit formula for the resolvent $G(z)=\int_{-a}^a{1\over \mu -u} \sigma(u)du$:
\[
G(z)={1\over \pi}\sqrt{z^2-a^2}\int_a^a{1\over \sqrt{u^2-a^2}}{\rm Log}\, {\al+1+u \over \al+1-u}{du \over z -u}
\]
The constraint on the density is equivalent to the
condition $G(z)={2\over z}+ O({1\over z^2)}$ as $z\to \infty$.
Solving this, we arrive to $a=\alpha+1-{\sqrt{3+4\alpha}\over 2}$.
Combining this with the explicit evaluation of the integral for $G(z)$ we obtain the limit density:
\begin{equation}\label{halfcutdensi}
\sigma(z)= {2\over \pi} \, {\rm Arctan}\, {\sqrt{3+4\al-4z^2}\over 1+2\al}
\end{equation}
for $z\in [0,\sqrt{{3\over 4}+\al}]$ and $\sigma(z)=0$ for $z>\sqrt{{3\over 4}+\al}$.

\subsubsection{} We could have  found this result a priori, by
recognizing that the half-cut hexagon from the previous section is the fundamental domain for the obvious reflection symmetry of the cut hexagon from the first section.

Because the limit density is the unique minimizer of the
rate functional, it is invariant with respect to this reflection.
In terms of the limit density or the cut hexagon from section
\ref{one} this is the symmetry $\rho(z)=\rho(\lambda+1-z)$.

Taking into account the change of variables from
Fig. \ref{tilingcuthex}  to Fig. \ref{fig:halfcut} we can see
that indeed
\[
\rho(z+\alpha +1)=\sigma(z)
\]
where $\lambda+1=2(\alpha +1)$, $\rho(t)$ is given by \eqref{uni-densit} and $\sigma(z)$ is by \eqref{halfcutdensi}.

\subsection{The reflection principle for limit densities}
\subsubsection{} The reflection principle which was illustrated above can be used to compute the limit density for a random tiling along the free boundary in a more general setting. In particular, it can be used to compute such density from the limit shape
on doubled domain with fixed boundary.

Assume that the free boundary of a domain $\mathcal D$ is a collection of segments on a line and  let ${\mathcal D}^r$ be the reflection with respect to this line.
The doubled domain ${\mathcal D}\cup {\mathcal D}^r$ is closed
with feasible (in a sense of \cite{KO}) boundary. It is invariant with respect to the reflection ${\mathcal D}\to {\mathcal D}^r$.
If we extend the definition of weights of tiling to
${\mathcal D}^r$ such that the weights will also be invariant
with respect to this reflection, the resulting limit shape
density will also be invariant.
It is also clear that in this case the
limit shape density for tilings of the domain ${\mathcal D}\cup {\mathcal D}^r$ at the symmetry axis is the same as
the the limit density for the tilings of ${\mathcal D}$ with the free boundary condition along this axis. Thus, the problem of computing the density of tiles at the free boundary reduces
to the computation of the limit shape for a symmetric domain with
fixed boundary along the symmetry axis.

As an example we can compare the limit density at the
diagonal of the symmetric hexagon with $\lambda=\theta$
with the density at the free boundary of the cut hexagon
with the same value of $\lambda$. The symmetry axis of the
hexagon correspond $x=\lambda$ in the computation of the density
along slices of the hexagon.
The density along this line is given by
eq.\eqref{densitone} with $x=\theta=\lambda$, and with $a$ and $b$ given by eq.\eqref{abhexa}.
These expressions for the hexagon case are identical to formulae \eqref{uni-bound}\eqref{uni-densit}for the density of endpoints of
lattice paths along the free boundary
for the half-hexagon of the same shape.

\begin{figure}
\centering
\includegraphics[width=15.cm]{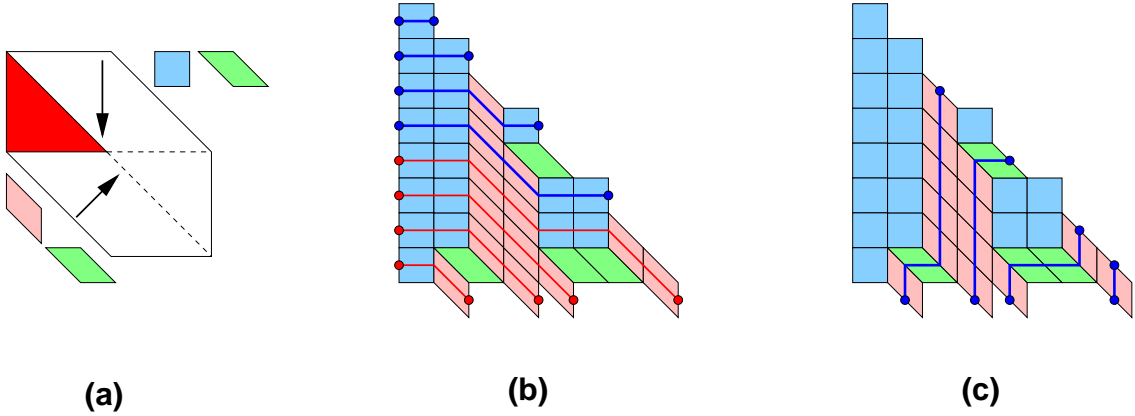}
\caption{The regular hexagon is decomposed into six triangles (a). The two inner boundaries
are both parts of lines that divide the hexagon in two (dashed lines), each corresponding to cuts along
lattice paths made of the two indicated tiles, for which the density along the cut is known.
In (b) we interpret the tiling of the triangle in terms of lattice paths starting at the fixed edge
and ending on the two free inner boundaries, along the third type of edge, not used for the paths in (a).
In (c), the same tiling is now interpreted differently in terms of paths that start on the lower free boundary
via tiles of the third type, and end on the top one with the two first type of tiles.}\label{fig:parallelo}
\end{figure}

\subsubsection{} Now let us discuss limit shapes corresponding to lattice paths in the cut hexagon with forbidden intervals along the cut.

Lattice paths and equivalent to tilings by rhombi which is also equivalent to monotonic piles
of unit cubes in the corner of a ``3D-room". This is
simply another way to phrase that rhombus tiling configurations are
equivalent to the corresponding height functions.

\begin{figure}
\centering
\includegraphics[width=12.cm]{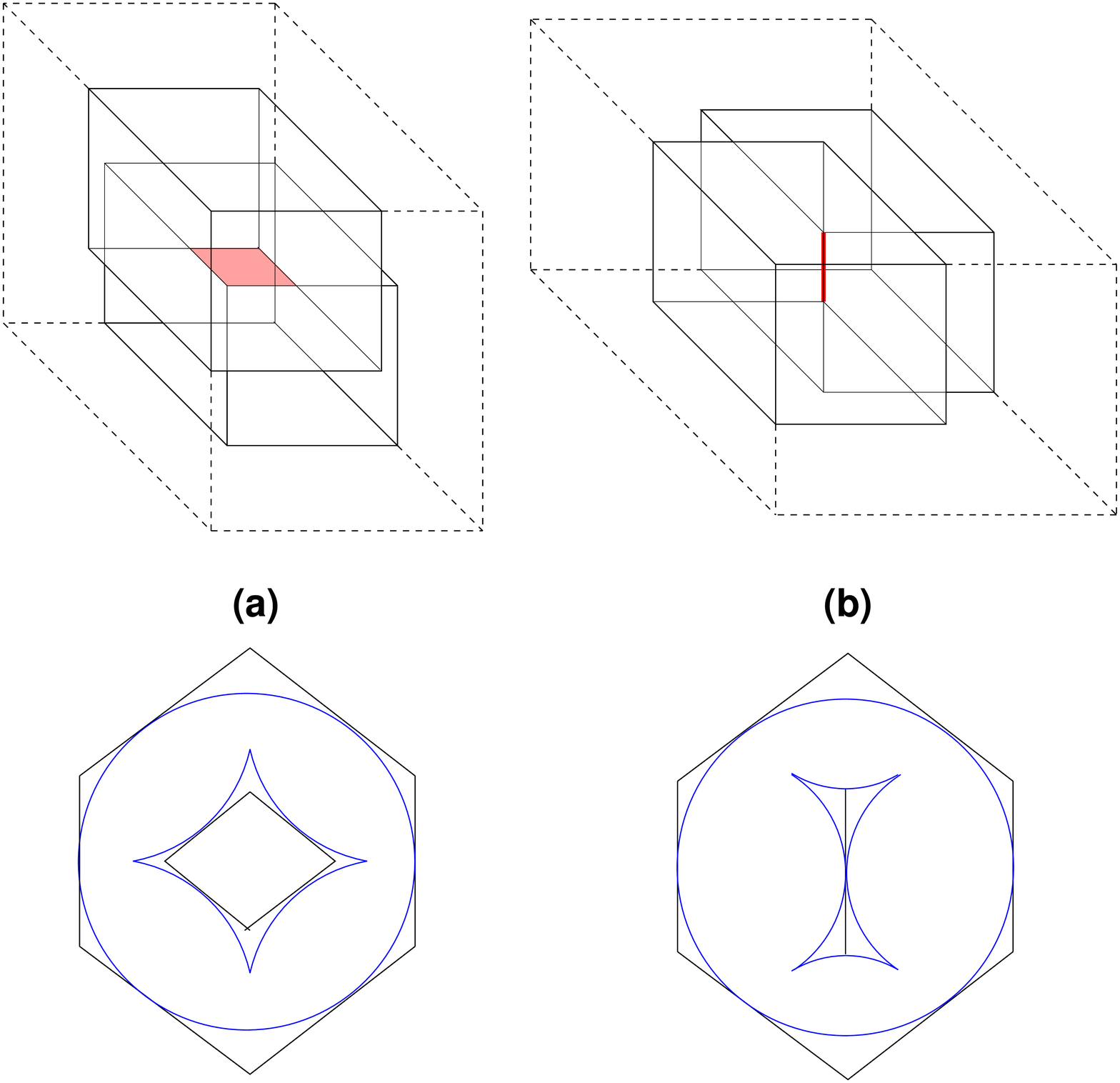}
\caption{The 3D realization of forbidden (a) and fully packed (b) boundaries. Both are obtained
by removing from the large parallelepiped (the room, in dashed lines) two smaller ones touching
opposite corners of the room. In the case (a), they share a plane domain
(a square in the case of regular objects), while in case (b) they intersect along a segment.
The diagonal of the square and the segment, in projection, correspond to respectively
the forbidden interval and the fully packed one. We have also represented the
expected limit shapes for the associated rhombus tilings.}\label{fig:paralepip}
\end{figure}

The case of the cut hexagon with one forbidden interval corresponds to pile of cubes in a rectangular room with two parallelepipeds
removed from it (see Fig.\ref{fig:paralepip} (a)).
One is removed from the upper closest corner, one is removed from the lower corner of the other side. They
should have the complementary heights (their sum is equal to the
height of the room) and they projections on the $(x,y)$ plane should intersect.
If this configuration is symmetric with respect to
the reflection $(x,y)\to (y,x)$, the lattice paths
corresponding to such a pile of cubes will have the forbidden interval along the diagonal of the square which is the
intersection of the lower side of the upper parallelepiped
and the upper side of the lower parallelepiped. If we
cut this rectangular room along the diagonal, the corresponding lattice paths
will be exactly lattice paths in the cut hexagon with
a forbidden interval along the cut.

In this case the density of the LGV paths along the cut is
determined by the limit slope of the corresponding tilings,
or equivalently by the limit height of corresponding piles
of cubes. The boundary of the limit shape in this case is disconnected. In terminology \cite{KO} this limit shape
is the result of the degeneration of the island with the hexagon inside and six cusps. The degeneration ``flattens"
the hexagon to a rhombus with horizontal main diagonal
(two vertical sides shrink to points). As a result six cusps
around the hexagonal island degenerate into four cusps around
the rhombus.

A similar interpretation holds for the density of end points of lattice paths with fully packed interval along the free
boundary of the cut hexagon. In this case piles of cubes
also should be in the a 3D rectangular room with two parallelepipeds removed
from it (see Fig.\ref{fig:paralepip} (b)). As in the case
of forbidden intervals, one of them is removed from the upper closest corner, one is removed from the lower corner of the other side. The difference is that their projections on the $(x,y)$ plane should be complementary (they should intersect over exactly one point). The sum of heights should be greater then the height
of the room. In this case parallelepipeds intersect over
a vertical segment. When such a configuration is symmetric with
respect to the reflection $(x,y)\to (y,x)$, the piles of cubes in such room with removed
parallelepipeds are equivalent to the
lattice paths on the left side continuing to the lattice paths on
the right side with the fully packed interval at the
diagonal cut (the interval where two parallelepipeds intersect).

In this case the limit shape will be discontinuous.
In the terminology of \cite{KO} this case is also a degeneration
of an island but the hexagon now becomes "squashed"
sidewise such that it degenerates to a vertical segment
(tilted sides degenerate to points).
In this case six cusps also degenerate to four cusps.

\subsubsection{} This symmetry principle extends to situations with other symmetries. Consider for instance
a regular hexagon (corresponding to $\lambda=\theta=1$ in the previous section).
Then the domain is invariant under rotations by $2\pi/3$, and reflections w.r.t. axes joining opposite vertices.
Let us cut the hexagon into
six triangles according to these symmetries, and focus on one of them (see Fig.\ref{fig:parallelo} (a)).

The counter clock-wise rotation by $2\pi/3$ brings a tiling
to a tiling. It maps the square looking tiles from the upper side of the triangle from Fig. 8(c)
to the tiles along the lower side of the triangle from Fig. 8(b)
where lattice paths end.

Because the hexagon is invariant with respect to rotations,
because the action function (the entropy, or the rate function) is invariant with respect to such rotations, and because it has unique minimizer the limit hight function is also invariant with respect
to such rotations, if one takes into account that the rotations also rotate tiles.

In particular , if $\sigma(z)$ is the density of lattice paths
along the lower boundary of the triangle from Fig. 8(c) and
$\rho(z)$ is the density of paths along the upper boundary of
the same triangle, we have $\rho(z)=1-\sigma(z)$. Here $z$ is the
distance from the right lower corner.

The density $\rho(z)$ we computed earlier in eq. \eqref{uni-densit} for all values of $\lambda$. For $\lambda=1$ we have:
\begin{equation}\label{exitile}
\rho(z)=\left\{\begin{matrix}
0 & z\in [0,1-\frac{\sqrt{3}}{2}] \\
{2\over \pi}\, {\rm Arctan}\, \sqrt{3-4(z-1)^2} & z\in[1-\frac{\sqrt{3}}{2},1] \end{matrix}\right.
\end{equation}
Here $z$ is the coordinate along the upper side of the triangle
with $z=0$ being the upper corner and $z=1$ being the lower right corner.

For the density $\sigma(z)$ we obtain
\begin{equation}\label{entertile}
\sigma(z)=\left\{\begin{matrix}
1 & z\in [0,1-\frac{\sqrt{3}}{2}] \\
{2\over \pi}\, {\rm Arctan}\, \sqrt{ 1\over \sqrt{3-4(z-1)^2}} & z\in[1-\frac{\sqrt{3}}{2},1] \end{matrix}\right.
\end{equation}

More generally, from the results for lattice paths in a hexagon, we obtain the density of the lattice paths in the triangle
from Fig. 8(c).  Let $L_x$ be the line parallel to the upper side of the triangle which is at the distance $x$
from the left lower corner. In the formulae below we use the
coordinate $z$ which is $0$ at the left side of the triangle and
changes from $0$ to $x$ along $L_x$ as we move downward.
Let $\rho(z,x)$ be the density of paths crossing $L_x$, $0\leq z\leq x$.

For $0\leq x\leq 1-\frac{\sqrt{3}}{2}$ we have
$\rho(z,x)=1$ for all $0\leq z\leq x$.

For $1-\frac{\sqrt{3}}{2}\leq x\leq \frac{1}{2}$
and $z\in [0,a]$
we have $\rho(z,x)=1$ and
\begin{eqnarray}
\rho(z,x)&=&1+{1\over \pi}\left( {\rm Arctan}\, \sqrt{b-z\over z-a}\sqrt{1+\theta-a\over 1+\theta-b}
-{\rm Arctan}\,  \sqrt{b-z\over z-a}\sqrt{1+x-a\over 1+x-b} \right. \nonumber \\
&&\qquad\qquad \left. +{\rm Arctan}\,  \sqrt{b-z\over z-a}\sqrt{a\over b}
-{\rm Arctan}\,  \sqrt{b-z\over z-a}\sqrt{\lambda-x+a\over \lambda-x+b}
\right)\label{densit1}
\end{eqnarray}
 when $z\in[a,x]$.
 Here  $a=\frac{1+x}{2}-\frac{1}{2}\sqrt{(2-x)3x}$ and $b=
 \frac{1+x}{2}+\frac{1}{2}\sqrt{(2-x)3x}$.

For $\frac{1}{2}\leq x \leq 1$, we have $\rho(z,x)=0$
when $0\leq z \leq a$ and
\begin{eqnarray}
\rho(z,x)&=&{1\over \pi}\left( {\rm Arctan}\, \sqrt{b-z\over z-a}\sqrt{1+\theta-a\over 1+\theta-b}
+{\rm Arctan}\,  \sqrt{b-z\over z-a}\sqrt{1+x-a\over 1+x-b} \right. \nonumber \\
&&\qquad\qquad \left. -{\rm Arctan}\,  \sqrt{b-z\over z-a}\sqrt{a\over b}
-{\rm Arctan}\,  \sqrt{b-z\over z-a}\sqrt{\lambda-x+a\over \lambda-x+b}
\right)\label{densit2}
\end{eqnarray}

These formulae compare with eq. \eqref{exitile} as $\rho(z)=\rho(z,1)$.

\subsubsection{Limit shape of TSSCPP}

\begin{figure}
\centering
\includegraphics[width=8.cm]{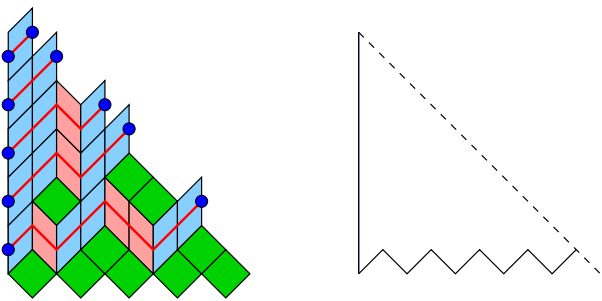}
\caption{TSSCPP are the rhombus tilings of a half-equilateral triangle, with fixed
boundary along the half-edge and the zig-zag median, and free boundary along the
remaining edge. We have indicated with dots the starting and ending points of the
non-intersecting paths with a floor, in bijection with the tilings. The exit tiles are those
marked by dots along the free boundary.
We expect the zig-zag line to become the reflection symmetry axis
of the full equilateral triangle in the continuum.}\label{fig:tsscpfin}
\end{figure}

The Totally Symmetric Self-Complementary Plane Partitions (TSSCPP) are in bijection
with the tilings of a half equilateral triangle of edge size $2n$, with fixed straight boundary
along one edge, fized zig-zag boundary along the median perpendicular to that edge, and
free boundary along the remaining edge (see Fig.\ref{fig:tsscpfin}).

The region in this problem can be naturally identified with the
fundamental domain of the full, fully symmetric hexagon.
It is natural to choose it as the triangle with $0\leq x\leq 1$ and $x\leq z\leq \frac{1+x}{2}$. Under this identification
the zig-gag line becomes the line $z=x$. As $n\to \infty$
it becomes the
symmetry axis of the limit shape of the height function for the
equilateral hexagon.

Using the same symmetry arguments as above we can compute the limit density of paths crossing the vertical line $L_x$ which is
at the distance $x$ from the left side of the triangle.
Denote this density by $\sigma(z,x)$.

When $0\leq x\leq 1-\frac{\sqrt{3}}{2}$ we have $\sigma(z,x)=1$ for
$0\leq z \leq a$ and $\sigma(z,x)$ is given by eq. \eqref{densit1}
for $a\leq z\leq \frac{1+x}{2}$.

When $1-\frac{\sqrt{3}}{2}\leq x\leq 1/2$ the density is given by
eq. \eqref{densit1} for all $x\leq z\leq \frac{1+x}{2}$. It is
given by eq. \eqref{densit1} for all $x\leq z\leq \frac{1+x}{2}$
when $1/2\leq x\leq 1$.

These expressions agree with the function eq. \eqref{entertile}:
$\sigma(x,x)=\sigma(x)$.

\subsection{Some other examples}

\begin{figure}
\centering
\includegraphics[width=10.cm]{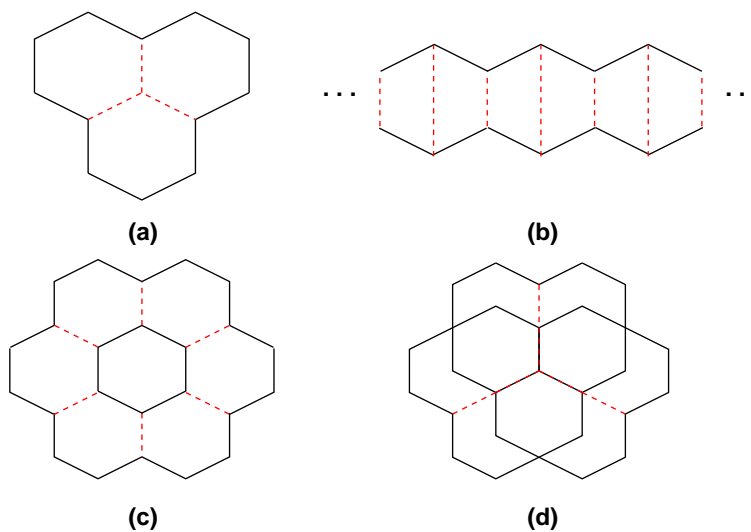}
\caption{Various applications of the symmetry principle, to a hexagon with
two free boundaries.}\label{fig:variedfree}
\end{figure}

Another example of the application of the symmetry principle
is the computation of the limit shape for the equilateral
hexagonal region with free boundary condition on two adjacent
sides. In this case one should glue three such hexagons into the
region shown on Fig.\ref{fig:variedfree}(a). The limit shape
for this region can be found using the methods of \cite{KO}.

The cut hexagon with both sides having
free boundary conditions on tilings correspond to a random
tiling of the infinite region obtained by gluing infinitely many
hexagons along vertical boundaries Fig.\ref{fig:variedfree}(b).
This case also describes the hexagon where two opposite sides
have the free boundary.

For more complicated examples of regions with
free boundaries the corresponding region with fixed boundary may
not be simply connected (see Fig.\ref{fig:variedfree}(c)),
nor planar (see Fig.\ref{fig:variedfree}(d)).

\end{document}